%% file: TSP_SRListDecoder.tex
\documentclass[lettersize,journal]{IEEEtran}

\usepackage[T1]{fontenc} %

\usepackage{textcomp}       %
\usepackage[utf8]{inputenc} %

\usepackage{stfloats}

\usepackage{amsthm}
\usepackage{mathtools}
\usepackage{amssymb}
\usepackage{amsfonts}
\usepackage{xfrac}
\usepackage{bm}

\newtheorem{Proposition}{Proposition}

\DeclareMathOperator*{\argmax}{argmax} %
\DeclareMathOperator*{\argmin}{argmin}

\DeclareMathOperator*{\HD}{HD}
\DeclareMathOperator*{\sgn}{sgn}

\newcommand{\bd}{\bm{d}}

\newcommand{\bv}{\bm{v}}

\newcommand{\psum}{\bm{\beta}}
\newcommand{\llr}{\bm{\lambda}}

\newcommand{\PM}{\mathrm{PM}}

\newcommand{\EV}{\mathrm{e}}
\newcommand{\OD}{\mathrm{o}}

\newcommand{\calN}{\mathcal{N}}

\newcommand{\bbA}{\mathbb{A}}
\newcommand{\bbF}{\mathbb{F}}
\newcommand{\bbS}{\mathbb{S}}

\usepackage{lipsum} %

\usepackage[usenames,dvipsnames]{xcolor} %

\usepackage{tikz}

\usepackage{pgfplots} %
\usepackage{multicol} %
\usepackage{graphicx} %
\usepackage{float}    %

\usepackage{subcaption} %

\usepackage{array}
\usepackage{booktabs}  %
\usepackage{multirow}
\usepackage{makecell}
\usepackage{boldline}
\usepackage{threeparttable}
\usepackage{pbox}

\usepackage{url}
\usepackage{verbatim}

\usepackage{algorithmic}
\usepackage[linesnumbered, lined, ruled, vlined]{algorithm2e}

\usepackage{cite}
\usepackage[super]{nth}
\usepackage[detect-weight=true]{siunitx}[=v2]

\usepackage[shortcuts]{extdash} %

\input{tikz/colorscheme.tex}

\IEEEoverridecommandlockouts

\newcommand{\tikzcircle}[2][black,fill=black]{\protect\tikz[baseline=-0.8ex]\protect\draw[#1,radius=#2, very thick] (0,0) circle ;}%
\newcommand{\tikzcirclewidth}{0.15}

\newcommand{\figwidth}{0.95} %

\newcommand{\mc}[3]{\multicolumn{#1}{#2}{#3}}

\newcommand{\pf}[1]{\textcolor{blue}{#1}} %
\renewcommand{\pf}[1]{#1} %

\begin{document}
\title{A Sequence Repetition Node-Based Successive Cancellation List Decoder for 5G Polar Codes: Algorithm and Implementation}

\author{%
  Yuqing Ren,~\IEEEmembership{Student Member,~IEEE,}
  Andreas Toftegaard Kristensen,~\IEEEmembership{Student Member,~IEEE,}
  Yifei Shen,\\~\IEEEmembership{Student Member,~IEEE,}
  Alexios Balatsoukas-Stimming,~\IEEEmembership{Member,~IEEE,}
  Chuan Zhang,~\IEEEmembership{Senior Member,~IEEE,}\\
  Andreas Burg,~\IEEEmembership{Senior Member,~IEEE}
  \vspace{-7mm}
  \thanks{Y. Ren, A. T. Kristensen, Y. Shen, and A. Burg are with the       Telecommunications Circuits Laboratory (TCL), \'{E}cole Polytechnique F\'{e}d\'{e}rale de Lausanne (EPFL), Lausanne 1015, Switzerland (email: \{yuqing.ren, andreas.kristensen, yifei.shen, andreas.burg\}@epfl.ch). \emph{Corresponding author: Andreas Burg}.}
  \thanks{Y. Shen, and C. Zhang are with the LEADS of Southeast University, the National Mobile Communications Research Laboratory, and the Purple Mountain Laboratories, Nanjing 210096, China (email: chzhang@seu.edu.cn).}
  \thanks{A. Balatsoukas-Stimming is with the Department of Electrical Engineering, Eindhoven University of Technology, 5600 MB Eindhoven, The Netherlands (email: a.k.balatsoukas.stimming@tue.nl).}
  \thanks{A. T. Kristensen and Y. Shen contributed equally to this paper.}
}

\maketitle

\begin{abstract}
Due to the low-latency and high-reliability requirements of 5G, low-complexity node-based successive cancellation list (SCL) decoding has received considerable attention for use in 5G communications systems.
By identifying special constituent codes in the decoding tree and immediately decoding these, node-based SCL decoding provides a significant reduction in decoding latency compared to conventional SCL decoding.
However, while there exists many types of nodes, the current node-based SCL decoders are limited by the lack of a more generalized node that can efficiently decode a larger number of different constituent codes to further reduce the decoding time.
In this paper, we extend a recent generalized node, the sequence repetition (SR) node\pf{,} to SCL decoding\pf{,} and describe the first implementation of an SR-List decoder.
By merging certain SR-List decoding operations and applying various optimizations for 5G New Radio (NR) polar codes, our optimized SR-List decoding algorithm increases the throughput by almost $\bm{2\times}$ compared to a similar state-of\=/the-art node-based SCL decoder.
We also present our hardware implementation of the optimized SR-List decoding algorithm which supports all 5G NR polar codes.
Synthesis results show that our SR-List decoder can achieve a \SI{2.94}{Gbps} throughput and \SI{6.70}{Gbps\per\milli\meter\squared} area efficiency for $\bm{L=8}$.
\end{abstract}

\begin{IEEEkeywords}
polar codes, successive cancellation decoding, list decoding, node-based SCL, sequence repetition (SR), hardware implementation, low-latency, 5G NR, wireless communications
\end{IEEEkeywords}

\section{Introduction}

\IEEEPARstart{P}{olar} codes, proposed in Ar{\i}kan's seminal work~\cite{arikan2009channel}, are the first class of error-correcting codes with an explicit construction that provably achieves channel capacity for binary-input discrete memoryless channels.
Based on their outstanding error-correcting performance, polar codes were ratified as the standard code for the control channel of 5G enhanced mobile broadband (eMBB)~\cite{5Gembb}.
While the polar code encoding schemes were released by 3GPP in 2018~\cite{5Gstandard2018}, 3GPP did not specify a decoding scheme for polar codes.
Although the original low-complexity successive cancellation (SC) decoding algorithm allows polar codes to achieve channel capacity at infinite code lengths, SC decoding has mediocre error-correcting performance for moderate code lengths and high decoding latency, which is insufficient for 5G scenarios.

To improve the error-correcting performance, Tal and Vardy proposed SC-List (SCL) decoding~\cite{tal2015list}, which explores both hypotheses for each message bit and maintains a list of up to $L$ candidate codewords in parallel during the decoding.
The 5G standard also includes concatenated cyclic redundancy check (CRC) codes, which allow the SCL decoder to select the most reliable candidate codeword that satisfies the CRC, which further improves the SCL decoding performance~\cite{Niu2012Letter}.
However, SCL decoding with medium to large list sizes (e.g., $L=4$ or $L=8$) suffers from high hardware complexity as an SCL decoder replicates the SC decoder structure $L$ times.
Nonetheless, owing to the hardware-friendly log-likelihood ratio (LLR) based SCL decoding algorithm~\cite{Alex2015LLRbased}, an SCL decoder with $L=8$ provides a good trade-off between error-correcting performance and hardware complexity.
It has therefore been chosen as the error-correcting performance baseline during the 5G standardization process~\cite{5GSCL}.

Driven by the strict low-latency and high throughput requirements of modern communication systems, numerous SCL decoder hardware implementations have been presented in the literature~\cite{yuan2014low, Lin2016TVLSI, Xiong2016TSP, Chui2018hroushold, liu20185, hashemi2016TCAS1,hashemi2017fastflexible, Kim2018TSP, Lee2020TSP, kestel2020506gbit, Alex2015LLRbased, Fan2016JSAC, Xiong2016TVLSI, Lin2015TVLSI, giard2017polarbear, Yuan2017TCASII, tao2020configurable}.
Most of these process more than one bit at a time to accelerate the decoding, since the conventional SCL decoder~\cite{Alex2015LLRbased} suffers from a long latency due to the serial decoding of each bit.
In~\cite{yuan2014low,Lin2016TVLSI, Xiong2016TSP,Chui2018hroushold,liu20185}, multi-bit SCL decoding is used in which a fixed number of bits are processed simultaneously, and all possible candidates for these bits are explored in parallel without exploiting a particular code structure.
Alternatively, due to the recursive construction of polar codes, where every polar code of length $N$ is formed from two constituent codes of length $N/2$, there exist constituent codes which have special bit patterns that allow for a more efficient direct decoding of all their bits in parallel.
Rate-0 (R0) and Rate-1 (R1) nodes were initially proposed in~\cite{alamdar2011simplified} for simplified SC decoding with many other node types later identified, such as repetition (REP) and single parity-check (SPC)~\cite{sarkis2014fast} nodes for Fast-SC decoding, TYPE-I to TYPE-V~\cite{hanif2017fast} nodes, and the generalized REP (G-REP) and generalized parity-check (G-PC) nodes~\cite{condo2018generalized}.
These node-based techniques were later extended with list update techniques for node-based SCL decoding~\cite{sarkis2015fast, hashemi2017fastflexible, hanif2018fast}.

These special nodes often represent large constituent codes, especially for high- and low-rate polar codes, allowing a node-based SCL decoder to \pf{generally} achieve \pf{much} lower latency than multi-bit SCL decoders \pf{at the cost of making the decoding schedule more complex due to the need for detecting these nodes in the decoding tree}.
However, all existing node-based SCL decoders~\cite{hashemi2016TCAS1, hashemi2017fastflexible, Kim2018TSP, Lee2020TSP, kestel2020506gbit} suffer from either only supporting a few special nodes, which impairs their decoding latency, or they instantiate separate sub-decoders for each node type when supporting several of these, which complicates their hardware implementation.
What is lacking is a more general node type implementation that supports most special node patterns with maximum hardware re-use and the ability to optimize the decoder throughput and area for a specific standard.

Recently, the sequence repetition (SR) node was proposed for SC decoding~\cite{zheng2020implementation, zheng2021threshold}.
The SR node is a generalized node that includes most aforementioned nodes as special cases.
An SR-based decoder can thus support many node patterns in both low- and high-rate codes to reduce the decoding effort without limiting the range of supported codes.
SR decoding can also potentially provide low-latency decoding as it is highly parallelizable, making it promising for SCL decoding in 5G.
However, for SCL decoding, with a list of $L$ paths, the strategy from~\cite{zheng2021threshold} of decoding with all the SR node candidates in parallel is highly computationally complex.
Therefore, to make SR-List decoding more feasible in hardware, it is critical to develop a general low-complexity SR-List decoding algorithm and architecture.
This architecture can then be customized to the specific codes employed by a standard to provide the best trade-off between area and decoding latency.

\subsection*{Contributions and Paper Outline:}
This work is an extension of our work in~\cite{shen2022}, in which the general SR-List \pf{decoding} algorithm was initially proposed.
However, due to the high implementation complexity when utilizing the large parallelism of SR-List \pf{decoding}, it is difficult to fully realize the potential of the SR-List \pf{decoding} algorithm in hardware.
In this paper, we thus perform a joint optimization of the SR-List \pf{decoding} algorithm with its hardware implementation.
\pf{To better reach a low-latency implementation, we optimize the SR node structure by dividing it into a low-rate SR-I part combined with a rate-one SR-II part, and perform their maximum likelihood (ML) decoding and Wagner-based list decoding by merging or simplifying internal decoding steps.}
We also analyze the 5G NR polar codes to derive several constraints on the special node types, which helps us to significantly decrease the decoding latency and computational complexity.
Our contributions on the algorithmic level comprise the following:
\begin{itemize}
    \item We propose an optimized SR-List decoding algorithm \pf{which makes} more operations execute in parallel \pf{and} less complex to \pf{implement} in hardware. Additionally, by investigating the node types encountered in the 5G NR polar codes, we can derive several constraints which further help to reduce the SR-List decoding complexity.

    \item For our optimized SR-List decoding, we also propose to combine it with two general latency reduction techniques which simplify the path update and node traversal, respectively. The worst-case latency is significantly reduced when compared to the state-of\=/the-art (SOA) node-based SCL~\cite{hashemi2017fastflexible} with negligible \pf{to no} performance degradation.
\end{itemize}

On the architecture level, \pf{we design the SR-List decoder implementation, which is the first node-based SCL decoder that} is fully compatible with and specifically optimized for all 5G NR polar codes.
Our contributions in this area include:
\begin{itemize}
    \item We propose flexible multi-stage decoding to combine multi-stage decoding~\cite{liu20185} and node-based list decoding~\cite{hashemi2017fastflexible}, which decreases the LLR memory and the LLR calculation latency.

    \item We provide an efficient architecture for SR-List decoding that efficiently re-uses internal modules to improve the hardware efficiency.

    \item Based on the rank\pf{-order} sorter~\cite{Gal2020SiPS}, we propose a \pf{new} partial rank\=/order sorter which has a \SI{45}{\percent} area reduction compared to a full rank\=/order sorter and improves the maximum clock frequency.

    \item By exploring the SR-List decoder design space, we find the best algorithm/architecture configurations that satisfy different requirements for area and throughput. Synthesis results show that our decoder can outperform other \pf{similar} polar decoders in throughput and area efficiency.
\end{itemize}

The remainder of this paper is organized as follows:
Section~\ref{sec:preliminaries} provides symbol definitions and a background on polar codes and decoding.
Section~\ref{sec:proposed_algos} introduces the proposed SR-List decoding algorithm \pf{with both general optimizations and} specifically for 5G NR polar codes.
In Section~\ref{sec:hw_node_based_polar_dec_arch}, we present our SR-List decoder architecture which is compatible with all 5G NR polar codes.
Section~\ref{sec:impl_results} discusses the implementation results and presents our design space exploration.
Finally, Section~\ref{sec:conclusion} concludes the paper.

\section{Preliminaries}\label{sec:preliminaries}

\emph{Notation:} In this paper, we use the following definitions.
Boldface lowercase letters $\bm{u}$ denote vectors, where ${\bm{u}}[i]$ refers to the $i$\=/th element of $\bm{u}$ and $\bm{u}[i{:}j]$ is the sub-vector $(\bm{u}[i],\bm{u}[i+1],\cdots,\bm{u}[j]),i\leq j$ and the null vector otherwise.
Boldface uppercase letters $\mathbf{B}$ represent matrices, with $\mathbf{B}[i][j]$ denoting the element in the $i$\=/th row of the $j$\=/th column.
Note that all indices start from $0$.
Blackboard letters $\bbS$ denote sets with $|\bbS|$ being the cardinality.
The function $\HD(x) := 1_{x < 0}$ defines the hard decision function.
We adopt the following parameters from the 5G NR standard~\cite{5Gstandard2018}:
the length of the mother polar code is denoted as $N=2^n$, $A$ is the number of message bits, $P$ the number of CRC bits, $K$ the number of messages bits with the CRC bits attached (i.e., information bits for polar codes), $E$ the length of the codeword after rate-matching, and $G$ the encoded block length.
The $i$\=/th node at stage $s$ is denoted as $\calN_{s,i}$, the code length of the node as $N_s$ or $2^s$, and the number of information bits in the node is $K_{s,i}$.
The frozen and information bit set indices are denoted as $\bbF$ and $\bbA$, respectively.
The code-rate $R$ is $R=A/E$, and we refer to a code after rate-matching as an $(E,A)$ polar code.
We use uplink (UL) and downlink (DL) codes from the Physical Uplink Control Channel (PUCCH), the Physical Uplink Shared Channel (PUSCH), and the Physical Downlink Control Channel (PDCCH) from 5G NR.
\pf{Further details on 5G NR are given in~\cite{5Gstandard2018} and the 5G NR polar code construction is described in~\cite{bioglio2020design}.}
\pf{A summary of key symbol and function definitions is provided in Table~\ref{tab:symbol_def} of Appendix~\ref{sec:symbol_def}.}

\subsection{Construction and Encoding}

Given an input bit sequence $\bm{u}$ of length $N$ with $A$ information bits, a polar encoder applies a linear transformation $\bm{x}=\bm{u}\mathbf{G}$ to get the codeword $\bm{x}$.
The generator matrix $\mathbf{G}=\mathbf{F}^{\otimes n}$ is constructed from the kernel $\mathbf{F}= \left[\begin{smallmatrix} 1 & 0 \\ 1 & 1\end{smallmatrix}\right]$.
Based on the principle of channel polarization~\cite{arikan2009channel}, the $N$ bits in $\bm{u}$ correspond to $N$ individual bit channels with different reliability, where the $K$ most reliable bit channels transmit information bits with CRC attached and the remaining $N-K$ bit channels transmit frozen bits, typically set to a value 0.
After encoding, rate-matching is applied to the codeword $\bm{x}$ to get a codeword of length $E$ for both UL and DL channels.

\subsection{Successive Cancellation (SC) Decoding}

At the receiver, a channel LLR vector of length $N$ is obtained after rate-recovery and $\bm{u}$ can be estimated by an SC decoder.
SC decoding can be represented as the traversal of a binary tree with $n+1$ stages (including the root node) and $2^{n-s}$ nodes at the $s$\=/th stage as illustrated in Fig.~\ref{fig:graph_dec_tree_ex_s4}.
Each node represents a constituent code of length $2^s$.
For the $i$\=/th node at the $s$\=/th stage, $\calN_{s,i}$ with $i\in [0,...,2^{n-s}-1]$, a length $2^s$ LLR vector $\llr_{s,i}$ is received and after traversing all the child nodes, the node returns a length $2^s$ partial sum (PSUM) vector $\psum_{s,i}$ to its parent node.
The $\llr_{s,i}$ update equations are
\begin{subequations}\label{eq:llr_fg_func}
    \begin{align}
        \llr_{s,2i}[j]   &= f(\llr_{s+1,i}[j],\llr_{s+1,i}[j+2^s]) \,, \label{eq:llr_fg_funcA} \\
        \llr_{s,2i+1}[j] &= g(\llr_{s+1,i}[j],\llr_{s+1,i}[j+2^s],\psum_{s,2i}[j]) \,, \label{eq:llr_fg_funcB}
    \end{align}
\end{subequations}
where $j \in [0, 2^{s}-1]$ indexes the $j$\=/th value of an LLR vector.
The $f$- and $g$-functions are defined as
\begin{subequations}\label{eq:fg_func_def}
    \begin{align}
        f(x,y)  & \approx \sgn(x) \sgn(y) \min\{|x|,|y|\} \,, \label{eq:fg_func_defA} \\
        g(x,y,z) &= (1-2z) x + y  \,. \label{eq:fg_func_defB}
    \end{align}
\end{subequations}
The PSUM vector, $\psum_{s+1,i}$, is updated by
\begin{subequations}\label{eq:psum_func}
    \begin{align}
        \psum_{s+1,i}[j]     &= \psum_{s,2i}[j]\oplus\psum_{s,2i+1}[j] \,, \label{eq:psum_funcA}\\
        \psum_{s+1,i}[j+2^s] &= \psum_{s,2i+1}[j] \,\pf{.} \label{eq:psum_funcB}
    \end{align}
\end{subequations}

\begin{figure}[t]
    \centering
    \includegraphics[width=\columnwidth]{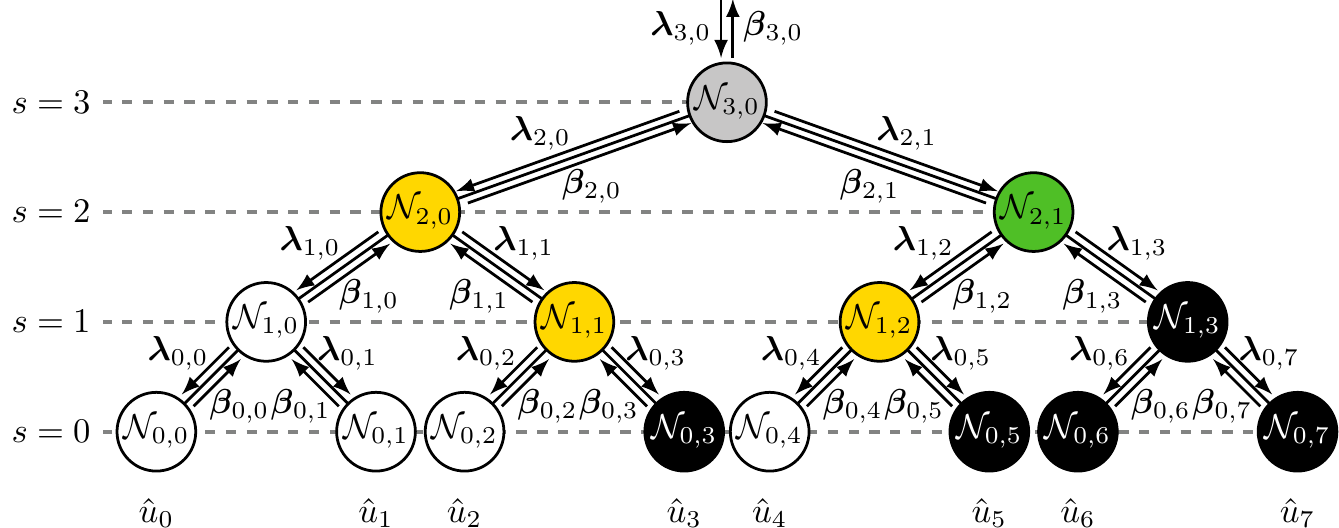}
    \caption{SC decoder tree for $N=8$ and $\bbF=\{0,1,2,4\}$. The node types are \tikzcircle[black, fill=\RZeroColor]{\tikzcirclewidth} for R0, \tikzcircle[black, fill=\ROneColor]{\tikzcirclewidth} for R1, \tikzcircle[black, fill=\RepColor]{\tikzcirclewidth} for REP, \tikzcircle[black, fill=\SpcColor]{\tikzcirclewidth} for SPC, and \tikzcircle[black, fill=shinygray]{\tikzcirclewidth} for a general node of any rate.}
    \label{fig:graph_dec_tree_ex_s4}
\end{figure}

\subsection{Successive Cancellation List (SCL) Decoding}

SC decoding only selects the locally optimal estimate for each information bit, and any wrong choice makes the entire estimated codeword incorrect.
Contrary to SC decoding, SCL decoding maintains a list of up to $L$ bit sequences (paths) by examining both hypotheses for each information bit~\cite{tal2015list}.
Each path is forked into two new paths.
The list is initialized with one path starting at the root.
Once the number of generated paths exceeds the list size, only the $L$ most reliable paths are kept by measuring their (approximate) path metrics (PMs) as
\begin{equation}\label{eq:PM_SCL_func}
    \PM_{0, i}^l \approx
    \begin{cases}
        \PM_{0, i-1}^l,                     & \text{if } \psum_{0,i}^{l}=\HD(\lambda_{0,i}^{l}) \,, \\
        \PM_{0, i-1}^l+|\lambda_{0,i}^{l}|, & \text{otherwise} \,, \\
    \end{cases}
\end{equation}
where $\lambda_{0,i}^{l}$ is the LLR of path $l$ at leaf node $\calN_{0, i}$, $\HD(\lambda_{0,i}^{l})$ is the hard decision based on $\lambda_{0,i}^{l}$, and $\psum_{0,i}^l$ is the $i$\=/th bit estimate for path $l$.

To differentiate between paths considered at each information bit, we refer to the $L$ paths before decoding an information bit as \emph{parent paths}, the $2L$ paths created from considering both hypotheses at an information bit as \emph{candidate paths}, and the $L$ selected paths with the smallest PMs as the \emph{surviving paths}.
When the next information bit is encountered, the previously selected paths are used as parent paths.
In~\eqref{eq:PM_SCL_func}, $\PM_{0, i-1}^l$ is the PM of a parent path and $\PM_{0, i}^l$ is the PM of a candidate path.

\subsection{Node-Based Decoding}\label{sec:node_based_decoding}

SC and SCL decoding suffer from long decoding latencies as the entire decoding tree is traversed.
However, some special types of nodes in the decoding tree have frozen bit patterns that allow for directly calculating the PSUMs using special decoding algorithms.
To label whether a leaf bit in the fanout tree of a node $\calN_{s,i}$ is frozen, we use a vector $\bd_{s,i}$, where $\bd_{s,i}[j]=0$ indicates a frozen bit and $\bd_{s,i}[j]=1$ indicates an information bit, respectively.
In the following, we omit the node index $i$ for brevity.

Using this notation, the initial four special nodes proposed in~\cite{alamdar2011simplified} and~\cite{sarkis2014fast} are:

\smallskip
\begin{tabular}{@{}p{6mm} @{}p{9mm} @{}l}
  1) & \textbf{R0}   & $\bd_s=(0,0,\cdots,0,0)$ \tabularnewline
  2) & \textbf{R1}   & $\bd_s=(1,1,\cdots,1,1)$ \tabularnewline
  3) & \textbf{REP}  & $\bd_s=(0,0,\cdots,0,1)$ \tabularnewline
  4) & \textbf{SPC}  & $\bd_s=(0,1,\cdots,1,1)$ \tabularnewline
\end{tabular}
\smallskip

Examples of these nodes are shown in Fig.~\ref{fig:graph_dec_tree_ex_s4}.
Note that some special nodes can be decomposed into other special nodes, e.g., the REP node $\calN_{2,0}$ contains both an R0 and an REP node.

The five additional nodes defined in~\cite{hanif2017fast} have the patterns:

\smallskip
\begin{tabular}{@{}p{6mm} @{}p{17mm} @{}l}
  5) & \textbf{TYPE-I}   & $\bd_s=(0,0,\cdots,0,0,0,1,1)$ \tabularnewline
  6) & \textbf{TYPE-II}  & $\bd_s=(0,0,\cdots,0,0,1,1,1)$ \tabularnewline
  7) & \textbf{TYPE-III} & $\bd_s=(0,0,1,1,\cdots,1,1,1)$ \tabularnewline
  8) & \textbf{TYPE-IV}  & $\bd_s=(0,0,0,1,\cdots,1,1,1)$ \tabularnewline
  9) & \textbf{TYPE-V}   & $\bd_s=(0,\cdots,0,1,0,1,1,1)$ \tabularnewline
\end{tabular}
\smallskip

\noindent
The top node in Fig.~\ref{fig:graph_dec_tree_ex_s4}, $\calN_{3,0}$, is an example of a TYPE-V node, which is composed of an REP node and an SPC node.

The G-REP and G-PC nodes were proposed in~\cite{condo2018generalized} and are more general nodes for low- and high-rate codes, respectively.
Their patterns are shown below, where $N_p$ indicates the code length of a node at stage $p$ with $p < s$ and each ``$\mathrm{X}$'' can either be a frozen or an information bit.

\vspace*{-0.75em}
\begin{tabular}{@{}p{6mm} @{\hskip 0.7mm}p{13mm} @{}l}
  10) & \textbf{G-REP} & $\bd_s=(0,0,\cdots,0,\overbrace{\mathrm{X},\cdots,\mathrm{X}}^{N_p})$ \tabularnewline
  11) & \textbf{G-PC} & $\bd_s=(\underbrace{0,\cdots,0}_{N_p},1,\cdots,1,1)$ \tabularnewline
\end{tabular}
\smallskip

Note that $N_p$ can be zero, meaning there are no $\mathrm{X}$s for the G-REP node and no $0$s for the G-PC node.
This new definition varies from~\cite{condo2018generalized}, but it allows the G-REP node to include the R0 node and the G-PC node to include the R1 node.

The more general SR node~\cite{zheng2021threshold}, shown in Fig.~\ref{fig:graph_sr}, is an extended form of the G-REP node, with multiple groups of R0/REP nodes and a source node located at stage $r$ that can be of any type.
The pattern of an SR node is expressed as:

\smallskip
\begin{tabular}{@{}p{6mm} @{\hskip 0.7mm}p{7mm} @{}l}
 12) & \textbf{SR} & $\bd_s=(\underbrace{0, \cdots, 0, \mathrm{X}}_{N_{s-1}}, \cdots, \underbrace{0, \cdots, 0, \mathrm{X}}_{N_{r+1}}, \underbrace{\mathrm{X}, \cdots, \mathrm{X}}_{N_r})$ \tabularnewline
\end{tabular}

\begin{figure}[t]
    \centering
    \includegraphics[width=\figwidth\columnwidth]{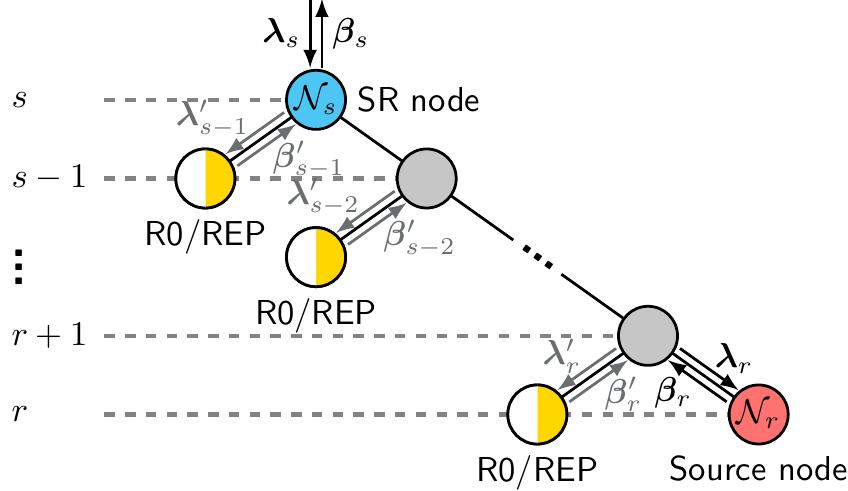}
    \caption{Tree representation of an SR node. \tikzcircle[black, fill=\SRColor]{\tikzcirclewidth} indicates the SR node and \tikzcircle[black, fill=\SourceColor]{\tikzcirclewidth} the source node of the SR node. The R0/REP node LLRs/PSUMs are indicated with prime $\prime$ as $\llr^{\prime}$ and $\psum^{\prime}$, respectively, and are grayed out as they are not calculated.}
    \label{fig:graph_sr}
\end{figure}

\subsection{Sequence Repetition (SR) Node Decoding}\label{sec:prelim_sr_node}

The SR node is a generalized node that covers most of the aforementioned nodes as special cases.
An SR-based decoder can thus support a wide range of frozen bit patterns.
An SR node, as shown in Fig.~\ref{fig:graph_sr}, is any node $\calN_s$ at stage $s$, whose descendants are all either R0 or REP nodes, except for the rightmost node at stage $r$, $0 \leq r \leq s$.
Node $\calN_r$ is called the \emph{source node}, and is a generic node of any rate.
The SR node structure is described by three parameters as $\mathrm{SR}(\bv, \mathrm{SNT}, r)$~\cite{zheng2021threshold}.
The length $(s-r)$ binary vector $\bv$ describes the distribution of R0 and REP nodes towards $\calN_r$.
For $0 \leq t < s-r$, $\bv[t] = 0$ and $\bv[t] = 1$ indicates that the left node at stage $(s-t-1)$ is an R0 node and an REP node, respectively.
As such, the sum of $\bv$ is the number of REP nodes, denoted as $W_v$.
The parameter SNT describes the source node type.
As an example, the node $\calN_{3,0}$ in Fig.~\ref{fig:graph_dec_tree_ex_s4} is an $\mathrm{SR}((1), \mathrm{SPC}, 2)$ node.
Additionally, we use a vector $\bm{\eta}$ to denote the last bit value (shown by an $\mathrm{X}$ in the special node description \pf{of the SR node}) of each R0/REP node, that is, $\bm{\eta}[t]=0$ when $\bv[t]=0$ and $\bm{\eta}[t]\in\{0,1\}$ when $\bv[t]=1$.

On a high-level, an SR node is decoded using different candidate LLR vectors, $\llr_r^k$, at the source node.
Each candidate LLR vector is generated as shown in~\eqref{eq:llr_source_node} from $\llr_s$ by repeatedly applying~\eqref{eq:fg_func_defB} using repetition sequences, $\bbS^k$, with each sequence generated from a case of $\bm{\eta}$~\cite{zheng2021threshold}.
Let $\bbS = \{\bbS^0, \bbS^1, \dots, \bbS^{2^{W_v}-1}\}$ denote the set of all possible repetition sequences for a given SR node.
For each $k$\=/th case of $\bm{\eta}$, the sequence $\bbS^k$ is derived as
\pf{\begin{equation}\label{eq:seq_rep_def}
    \bbS^k = (\bm{\eta}[0],0)\boxplus(\bm{\eta}[1],0)\boxplus\cdots\boxplus(\bm{\eta}[s-r-1],0) \,,
\end{equation}}
with $\boxplus$ describing the operation
\pf{\begin{equation}\label{eq:seq_rep_op_def}
    \begin{aligned}
         &(\bm{a}[0], \bm{a}[1], \cdots, \bm{a}[i]) \boxplus (\bm{b}[0], \bm{b}[1]) \\
        = &(\bm{a}[0]{+}\bm{b}[0], \bm{a}[0]{+}\bm{b}[1], \cdots, \bm{a}[i]{+}\bm{b}[0], \bm{a}[i]{+}\bm{b}[1]) \, .
    \end{aligned}
\end{equation}}

The SR node decoding process can then be divided into three main steps as:

1) The $|\bbS|$ LLR vectors to the source node are calculated as
\begin{equation}\label{eq:llr_source_node}
    \llr_{r}^k[j]=\sum\limits_{m=0}^{2^{s-r}-1}\left(1-2 \bbS^k[m]\right) \llr_{s}[2^r m +j] \,,
\end{equation}
where $0\leq j<2^{r}$ and $0\leq k<|\bbS|$.

2a) Given the $|\bbS|$ LLR vectors, the source node is decoded using $|\bbS|$ copies of the conventional SC decoder in parallel.

2b) Meanwhile, the \pf{ML} candidate of the $|\bbS|$ candidates is identified as
\begin{equation}\label{eq:psum_sr_ml_idx}
    \widehat{k}=\argmax_{0\leq k<|\bbS|}\sum_{j=0}^{2^r-1}|\llr_{r}^k[j]| \,.
\end{equation}

3) The SR node PSUM vector, $\psum_s$, is computed by repeating $\psum_{r}^{\widehat{k}}$, based on the corresponding $\widehat{k}$\=/th repetition sequence
\begin{equation}\label{eq:psum_sr}
    \psum_{s}[2^r m + j]=\psum_{r}^{\widehat{k}}[j]\oplus \bbS^{\widehat{k}}[m] \,.
\end{equation}
where $0\leq j<2^{r}$ and $0\leq m<2^{s-r}$.

Node-based decoding techniques can be used in step~2a when the source node is a special node and step 2b can be skipped when there are only R0 nodes as left-descendants \pf{since} $|\bbS|=1$.
Note that the LLRs for all the left R0/REP nodes under the SR node, $\llr^{\prime}_{s-1}, \llr^{\prime}_{s-2}, \dots, \llr^{\prime}_{r}$, are never calculated and are grayed out in Fig.~\ref{fig:graph_sr} along with the PSUMs.

\section{Proposed \pf{SR-}List Decoding Algorithms}\label{sec:proposed_algos}

In this section, we propose several node-based list decoding algorithms for a low-latency SCL decoder using a single generalized node, the SR node.
First, we present a general SR-List decoding algorithm to extend the SR node for list decoding with no error-correcting performance loss.
Then, we propose an optimized low-complexity SR-List decoding algorithm based on the 5G NR node distribution and the merging or simplification of some SR-List decoding operations.
We also provide a strategy for empirically reducing the number of path forks for SR-List decoding and we show how some 5G NR rate-matching schemes can reduce the decoding latency without any loss in error-correcting performance.

\subsection{General SR-List Decoding}\label{sec:general_SRL}

As described in Section~\ref{sec:prelim_sr_node}, the SR node provides a high degree of parallelization by considering all repetition sequences in parallel.
While this can provide a reduction in decoding latency compared to other node types, all parent paths and their path forks have to be considered for all repetition sequences during list decoding.
For an SR node with a source node of size $N_r$ and information length $K_r$, the total number of information bits is $W_v + K_r$, and the $L$ parent paths can each be forked $|\bbS|2^{K_r}$ ($|\bbS|=2^{W_v}$) times to account for all choices of information bits,
which significantly increases the decoding and sorting complexity.

To alleviate this issue, we simplify the SR-List decoding process by dividing it into two parts, described in Section~\ref{sec:list_dec_r0_rep} and Section~\ref{sec:list_dec_source_node}, respectively.
In the first part, for the R0 and REP nodes, we enumerate all information bit possibilities to generate $|\bbS|$ candidates from each of the $L$ parent paths, calculate the PMs of all $|\bbS|L$ candidate paths, and select $L$ survivors.
Then, in the second part, sequential node-based list decoding is performed for the source node using the previously selected $L$ candidates.
Note that in the hardware, we simplify the decoding for basic nodes such as R0, R1, REP, SPC, and TYPE-III as explained in Section~\ref{sec:hw_rsu}.

\subsection{General SR-List Decoding Part 1: R0/REP Node Decoding}\label{sec:list_dec_r0_rep}

In conventional node-based list decoding such as Fast-SCL decoding~\cite{hashemi2017fastflexible}, SR node decoding starts by decoding each of the left-descendant R0 and REP nodes in order, and the PMs for the $L$ parent paths at the source node are based on the input LLRs of all the individual R0 and REP nodes.
However, in SR-List decoding, these LLRs are unknown to the decoder.
Instead, we propose a four step process for efficiently calculating the PMs for $L$ parent paths at the source nodes:

1) The LLR vectors at the source node, $\llr_r^{l,k}$, can be calculated for each path as
\begin{equation}\label{eq:llr_source_node_SCL}
    \llr_r^{l,k}[j]=\sum\limits_{m=0}^{2^{s-r}-1}\left(1-2\bbS^k[m]\right) \llr_s^{l}[2^r m+j] \,,
\end{equation}
where $0 \leq l<L$, $0 \leq k<|\bbS|$, and $0\leq j<2^{r}$.

2) To calculate the corresponding PMs of the source node, we propose a method which does not directly decode the R0 and REP nodes.
Let $\widetilde{\psum}_r^{l,k} = \mathrm{HD}(\llr_{r}^{l,k})$, to distinguish from $\psum_r^{l,k}$ which is the source node output after decoding it.
Then, we calculate $\widetilde{\psum}_{s}^{l,k}$ as
\begin{equation}\label{eq:psum_sr_est_SCL}
    \widetilde{\psum}_{s}^{l,k}[2^r m + j]=\widetilde{\psum}_{r}^{l,k}[j] \oplus \bbS^k[m] \,, \\
\end{equation}
where $0\leq j<2^{r}$ and $0\leq m<2^{s-r}$.

3) With the hard decision PSUMs at the source node and SR node, we then calculate $\PM_r^{l,k}$ at the source node for each sequence and for each parent path as
\begin{equation}\label{eq:PM_source_node_SCL}
    \PM_r^{l,k} =\PM_s^l+\sum\limits_{j=0}^{2^s-1}| \widetilde{\psum}_s^{l}[j] - \widetilde{\psum}_{s}^{l,k}[j]| |\llr_{s}^{l}[j]| \,,
\end{equation}
where $\PM_s^l$ indicates the PM of the $l$\=/th parent path at the SR node.
Note that $\PM_r^{l,k}$ is identical with the results obtained from conventional Fast-SCL decoding where the LLRs and PSUMs from the R0 and REP nodes are \pf{calculated}.

4) With $\PM_r^{l,k}$, we select the $L$ best of the $|\bbS|L$ \pf{candidate paths} and use these $L$ paths for the source node list decoding.

Note that for each of the $L$ paths that survived, the selected repetition sequence is recorded as $\widehat{k}_l$ for $0 \leq l < L$.
In the next part of the general SR-List decoding, described in Section~\ref{sec:list_dec_source_node}, $l$ refers to a candidate that survived in step~4 above, and not to the parent paths at the SR node input.

\subsection{General SR-List Decoding Part 2: Source Node Decoding}\label{sec:list_dec_source_node}

With the $L$ surviving paths as the parent paths for the source node, the decoder can start the source node list decoding.
Unfortunately, the source node of the SR node can be any type of node, which significantly complicates the hardware implementation.
However, if we constrain the source node to a G-PC node, whose first $N_p$ bits are frozen bits as shown in Fig.~\ref{fig:graph_sr_optimized}, the source node can be viewed as a group of $N_p$ SPC nodes.
\pf{In Lemma 2 of~\cite{condo2018generalized}, the list decoding strategy for the G-PC node is presented where $N_p$ SPC nodes are decoded in parallel.}
However, such a parallel strategy explores many redundant candidates.
Therefore, inspired by the SPC and TYPE-III list decoding processes in~\cite{hashemi2017fastflexible} and~\cite{hanif2018fast}, respectively, we propose an efficient sequential list decoding algorithm for G-PC nodes.

1) For each path, split the LLRs and PSUMs into $N_p$ groups and find each group's minimum magnitude LLR index as
\begin{equation}\label{eq:gpc_index_SCL}
    \epsilon_q^{l}=\argmin\limits_{0 \leq i < N_r / N_p} |\llr_r^l[N_p i + q]| \,,
\end{equation}
where $0\leq q<N_{p}$ is the group index.

2) For each path and group, compute the initial parity
\begin{equation}\label{eq:gpc_parity_SCL}
    \gamma_q^l=\bigoplus\limits_{i=0}^{N_r/N_p-1} \widetilde{\psum}^l_{r}[N_p i+q] \,,
\end{equation}
with $\widetilde{\psum}^l_{r}$ obtained at the source node after calculating the source node LLRs using~\eqref{eq:llr_source_node_SCL}  as described in Section~\ref{sec:list_dec_r0_rep}.

3) For each path, generate the ML candidate by computing the PSUMs, $\widehat{\psum}_{r}^{l}$, which satisfy the even-parity constraint of each group as shown in~\eqref{eq:psum_gpc_ML_SCL} and calculate the corresponding candidate PM, $\widehat{\PM}^l_{r}$, as shown in~\eqref{eq:PM_gpc_ML_SCL}
\begin{align}
    \widehat{\psum}^l_{r}[N_p i+q] &=
    \begin{cases}
        \widetilde{\psum}^l_{r}[N_p i+q]\oplus \gamma_q^l, & \text{if } i=\epsilon_q^l \,, \\
        \widetilde{\psum}^l_{r}[N_p i+q], & \text{otherwise} \,.
    \end{cases}\label{eq:psum_gpc_ML_SCL} \\
    \widehat{\PM}^l_{r} &=\PM_r^l + \sum\limits_{q=0}^{N_p-1} \gamma_q^l  |\llr_r^l[N_p \epsilon_q^{l} + q]| \,, \label{eq:PM_gpc_ML_SCL}
\end{align}
with $\PM^l_{r}$ as the PM of a candidate allowed to survive after decoding the R0 and REP nodes, as described in Section~\ref{sec:list_dec_r0_rep}.

4) To reduce the exploration of redundant candidates, we modify the LLRs of each sub-group as follows.
Let ${\bm{\zeta}}_{r}^{l} = \llr_r^l$ except for $K_{r} = N_{r} - N_{p}$ bits which are defined as
\begin{equation}\label{eq:llr_gpc_modified_SCL}
    \bm{\zeta}_r^l[N_p i + q] = |\llr_r^l[N_p i + q]| + (1 - 2\gamma_q^l) |\llr_r^l[N_p \epsilon_q^{l} + q]| \,,
\end{equation}
where $0\leq i<N_r/N_p$ and $i\neq \epsilon_q^{l}$.

5) We then identify the bit-flipping candidates for further path forking by sorting $|\bm{\zeta}_r^l|$ in ascending order with sorted index $j^{l}$, i.e., $j^{l}$ is the index of the $j$\=/th minimum in $|\bm{\zeta}_r^l|$.
Note that the LLRs found in~\eqref{eq:gpc_index_SCL} are excluded from the sorting as these are used to satisfy the parity constraint of all candidates.

6) Then, we serially fork the paths from sorted index $0^{l}$ to $(K_{r}-1)^l$ for the $l$ paths, that is, the number of path fork operations is $K_{r}$.
The PMs for the generated codewords are
\begin{equation}\label{eq:PM_gpc_SCL}
    \widehat{\PM}_{r, j}^l =
    \begin{cases}
        \widehat{\PM}_{r, j-1}^l, & \text{if } \widehat{\psum}_r^l[j^{l}]=\widetilde{\psum}^l_r[j^{l}] \,, \\
        \widehat{\PM}_{r, j-1}^l+\Delta_{r,j}^l, &\text{otherwise} \,.
    \end{cases}
\end{equation}
where $\widehat{\PM}_{r, -1}^l=\widehat{\PM}_r^l$ and $\Delta_{r,j}^l$ is given as
\begin{equation}\label{eq:Delta_PM_gpc_SCL}
    \Delta_{r,j}^l=|\bm{\llr}_r^l[j^{l}]| + (1-2\gamma^l_{q_{j}^{l}}) \pf{|\llr^l_r [N_p \epsilon^l_{q_{j}^{l}} +q_{j}^{l}]|} \, ,
\end{equation}
with $q_{j}^{l} = j^{l} \mod N_p$, i.e., the group index of the LLR with sorting index $j^{l}$.
Note that if $\widehat{\PM}_{r, j-1}^l$ is incremented with $\Delta_{r,j}^l$, the corresponding $\gamma^l_{q_{j}^{l}}$ should be flipped.
The $L$ paths with the lowest PMs are selected for the next path forking and this process is repeated $\min(L-1,K_{r})$-times.

7) Finally, the $L$ SR node PSUMs can be obtained as
\begin{equation}\label{eq:psum_sr_final_SCL}
    {\psum}_{s}^{l}[2^r m + j]={\widehat{\psum}}_{r}^{l}[j] \oplus \bbS^{\widehat{k}_l}[m] \,, \\
\end{equation}
where $0\leq j<2^{r}$, $0\leq m<2^{s-r}$, and $\widehat{k}_{l}$ is the $\widehat{k}$\=/th sequence used for path $l$ found in Sec~\ref{sec:list_dec_r0_rep}.

\begin{figure}[t]
  \centering
  \includegraphics[width=\figwidth\columnwidth]{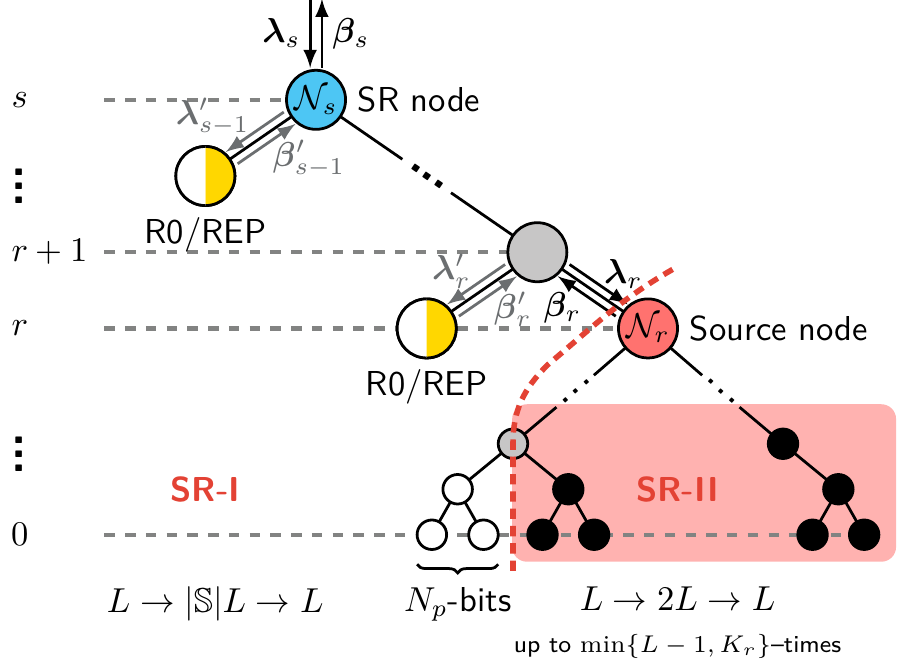}
  \caption{Optimized SR representation with SR\=/I and SR\=/II parts.}
  \label{fig:graph_sr_optimized}
\end{figure}

Modifying the LLRs in~\eqref{eq:llr_gpc_modified_SCL} results in a set of LLRs where the sub-groups with a parity of $1$ (i.e., $\gamma_q^l = 1$ in~\eqref{eq:gpc_parity_SCL}) are penalized and all LLR magnitudes in that group are reduced, making it more likely that those indices will be amongst the first to be explored, thus reducing the number of redundant candidates compared to~\cite{condo2018generalized}.
Note that if the G-PC node is an R1 node, we skip steps~1--4 and directly perform the path forking.
If the G-PC node is an SPC node, the LLRs in~\eqref{eq:llr_gpc_modified_SCL} are unnecessary and the process simplifies to SPC list decoding.

\subsection{Optimized SR-List Decoding}\label{sec:opt_sr_list_dec}

For the overall SR-List decoding, the source node list decoding waits until the list decoding of the R0/REP groups is complete.
To reduce the time complexity of this process, we propose an optimized SR-List decoding algorithm where we redefine the SR node representation and merge some SR-List decoding steps.
Moreover, by taking the specific 5G NR node distribution into account, we apply some constraints on the special nodes to significantly simplify the implementation.

\subsubsection{Optimized SR Node Representation}\label{sec:opt_sr_node_representation}

As described in Section~\ref{sec:list_dec_source_node}, the G-PC node is a good generalized form of the source node in the SR node, whose frozen bits are distributed in the front and thus do not cause path forking.
If we only use the G-PC node as the source node type, the decoding can be simplified by considering the source node parity check together with the list decoding of the R0 and REP segments in the SR node.\footnote{This constraint does not prevent the decoder from decoding any arbitrary code as code segments with an SR node incompatible with a G-PC source node can still be decoded using conventional decoding steps based on basic nodes. The impact on the time complexity then depends on the \pf{rate of occurrence} of such incompatible SR nodes, \pf{which are in general rare}.}
With this approach, the SR node is divided into two parts, SR\=/I (a low-rate part) and SR\=/II (a rate-one part), as shown in Fig.~\ref{fig:graph_sr_optimized}, with patterns

\begin{algorithm}[t]
    \small
    \SetKwInOut{Input}{Input}\SetKwInOut{Output}{Output}
    \caption{\texttt{Proposed SR-List Decoding}}
    \label{alg:Optimized_SR_List}
    \Input{$\llr_s^{l}\;\PM_s^l,\;\bbS$}
    \Output{$\psum_s^l$}
    \tcp{SR\=/I list decoding}
    \For{\upshape $l=0\;\textbf{to}\;L-1$}
    {
        \For{$k=0\;\textbf{to}\;|\bbS|-1$}
        {
            $\texttt{Calculate }\llr_r^{l,k}\texttt{ according to }\eqref{eq:llr_source_node_SCL}$\;
            $\texttt{Calculate }\widetilde{\psum}_{s}^{l,k}\texttt{ according to }\eqref{eq:psum_sr_est_SCL}$\;
            $\texttt{Calculate }\widehat{\PM}_r^{l,k}\texttt{ according to }\eqref{eq:PM_gpc_opt_SCL}$\;
        }
    }
    $\texttt{Select L paths from the }L\texttt{ smallest }\widehat{\PM}_r^{l,k}$\;
    \tcp{SR\=/II list decoding}
    \For{\upshape $j=0\;\textbf{to}\;\min(L-1,K_{s})$}
    {
        \For{\upshape $l=0\;\textbf{to}\;L-1$}
        {
            $\texttt{Calculate }\mathrm{PM}_{r, j}^{l}\texttt{ according to }\eqref{eq:PM_gpc_SCL}$\;
        }
    }
    $\texttt{Generate }\psum_{s}^l \texttt{ according to }\eqref{eq:psum_sr_final_SCL}$\;
    $\texttt{Return }\psum_{s}^l \texttt{ to the parent node}$\;
\end{algorithm}

\smallskip
\begin{tabular}{@{}p{6mm} @{}p{10mm} @{}l}
   13) & \textbf{SR\=/I}  & $\bd_s=(\overbrace{0,\cdots,0,\mathrm{X}}^{N_{s-1}}, \cdots, \overbrace{0,\cdots,0,\mathrm{X}}^{N_r}, \overbrace{0,\cdots,0}^{N_p})$ \tabularnewline
   14) & \textbf{SR\=/II} & $\bd_s = (\underbrace{\mathrm{1},\cdots,\mathrm{1}}_{N_r-N_{p}})$ \tabularnewline
\end{tabular}
\smallskip
where $N_p\in\{0,1, 2, 4,\dots\}$.

For SR\=/I, i.e., the group of R0 and REP nodes with the frozen bits from the G-PC node, the PMs for the $L$ paths that survived (after SR\=/I), denoted as $\widehat{\PM}_{r}^{l,k}$, can be derived by combining~\eqref{eq:PM_source_node_SCL} with~\eqref{eq:PM_gpc_ML_SCL} as follows
\begin{equation}\label{eq:PM_SRI_SCL}
    \begin{aligned}
        \widehat{\PM}_{r}^{l,k} = \PM_{s}^{l}&+\sum\limits_{j=0}^{2^s-1}|\widetilde{\psum}_{s}^{l}[j]- \widetilde{\psum}_{s}^{l,k}[j]||\llr_{s}^{l}[j]| \\
       &+\underbrace{\sum\limits_{q=0}^{N_p-1} \gamma_q^{l,k}  |\llr_r^{l,k}[\epsilon_q^{l,k} N_p+q]|}_{\pf{\Delta^{l,k}}} \,,
    \end{aligned}
\end{equation}
\pf{where ${\Delta^{l,k}}$ denotes the PM penalty from the parity constraints of the source node (G-PC node) for path $l$ with repetition sequence $k$.}
Note that all terms in~\eqref{eq:PM_SRI_SCL} can be calculated in parallel after~\eqref{eq:llr_source_node_SCL}.
Moreover, the best $L$ paths that survive from the $|\bbS|L$ paths of~\eqref{eq:PM_SRI_SCL} slightly improve the error-correcting compared to~\eqref{eq:PM_source_node_SCL} (without considering the parity), due to the more global (direct) selection of the $L$ best paths from the $|\bbS|L$ candidate paths.

The full SR-List decoding algorithm is given in Algorithm~\ref{alg:Optimized_SR_List}.
Note that in Algorithm~\ref{alg:Optimized_SR_List} we use~\eqref{eq:PM_gpc_opt_SCL} which we introduce in Section~\ref{sec:opt_sr_list_dec} as a simplification of~\eqref{eq:PM_SRI_SCL}.
In contrast to the SR\=/I part, the SR\=/II list decoding is performed by sequentially selecting $L$ paths from $2L$ generated candidate paths.
After list decoding the source node, the PSUM vectors to the SR parent node are calculated according to~\eqref{eq:psum_sr_final_SCL}.
In the remaining sections of this paper, we only use this optimized SR node representation.

\subsubsection{SR and G-PC node distribution in 5G NR}\label{sec:investigation_sr_gpc_5g}

For 5G NR polar codes, we can potentially find some constraints on the SR and G-PC node parameters which simplify the hardware \pf{with} little to no increase in worst-case decoding latency.
In particular, we note that more complex, potentially unsupported SR or G-PC nodes can still be decomposed into other nodes, which require slightly more time to decode, but where the \pf{reduced decoding} time does not justify the additional hardware resources.
To find \pf{constraints which reduce the overall complexity}, we check all valid \pf{5G NR polar} code configurations, \pf{i.e., combinations of $A$ and $G$.}
For PDCCH we use $A\in[12,140]$ and $G\in[36,8192]$.
\pf{For PUCCH/PUSCH we use $A\in[12,1706]$ and all the supported $G$ as shown below (taken from (1) in~\cite{egilmez2019development})}
\pf{\begin{equation}
G \in \begin{cases}
  [A + 9, 8192]                     & \text{if } A \in [12, 19] \\
  [A + 11, 8192]                    & \text{if } A \in [20, 359] \\
  [A + 11, 16385]                   & \text{if } A \in [360, 1012] \\
  [2 \lceil A / 2 \rceil+22, 16385] & \text{if } A \in [1013, 1706]. \\
\end{cases}
\end{equation}}

In Table~\ref{tab:5G_gpc}, we show the percentages of the different SR \pf{node} sizes and G-PC node types \pf{depending on the number of repetition sequences $|\bbS|$ in the SR node and the number of frozen bits $N_p$ in the G-PC node.}
For SR nodes in PDCCH, the number of SR nodes with \pf{$|\bbS|\leq4$ and $|\bbS|\leq8$} accounts for \pf{\SI{92.16}{\percent} and \SI{98.12}{\percent}}, respectively.
For PUCCH and PUSCH, the number of SR nodes with \pf{$|\bbS|\leq4$ and $|\bbS|\leq8$} is \pf{\SI{95.7}{\percent} and \SI{99.03}{\percent}, respectively.}
Additionally, we note that SR nodes with $|\bbS|>32$ never occur.
As described in Section~\ref{sec:general_SRL}, the computational complexity of decoding the SR node part with the R0 and REP nodes is linearly related to $|\bbS|L$.
Therefore, as SR nodes with large $|\bbS|$ rarely occur, we only consider $|\bbS|_{\mathrm{max}} \in \{2, 4, 8\}$ as interesting options for further evaluation in the hardware implementation.

\begin{table}[t]
    \captionsetup{font=small, justification=centering}
    \centering
    \caption{\MakeUppercase{SR and G-PC node distribution in 5G NR}}
    \label{tab:5G_gpc}

    \small

    \begin{tabular}{lrrr}
        \toprule
        \textbf{Name} & \textbf{Size} & \textbf{PDCCH} $[\%]$ & {\textbf{PUCCH/}\textbf{PUSCH}} $[\%]$ \\ \midrule
        \multirow{6}{*}{$\mathbb{|S|}$} & $1$ & $62.78$ & $68.85$             \\
        & $2$           & $19.07$              & $19.12$             \\
        & $4$           & $10.31$              & $7.73$              \\
        & $8$           & $5.96$               & $3.33$              \\
        & $16$          & $1.88$               & $0.93$              \\
        & $32$          & $4.31\times10^{-4}$  & $4.28\times10^{-2}$ \\ \midrule
        \multirow{3}{*}{$N_{p}$} & $0$  & $8.47$ & $17.42$   \\
        & $1$           & $83.91$              & $76.31$             \\
        & $2$           & $7.62$               & $6.27$              \\ \bottomrule
    \end{tabular}
\end{table}

For the G-PC node in PDCCH and PUCCH/PUSCH, we only observe $N_p\in \{0, 1, 2\}$, where $N_p = 0$ corresponds to having a R1 node, $N_p = 1$ an SPC node, and $N_p = 2$ a TYPE-III node.
G-PC nodes with $N_p>2$ never appear.
Therefore, even though the G-PC node is a generalized node, only a limited number of cases occur in valid 5G NR codes and we therefore restrict our decoder to $N_{p_\mathrm{max}}=2$.

Consequently, in the following, all simulated SR-List decoding results are based on the optimized SR-List decoding with the constraints $|\bbS|_{\mathrm{max}}=8$ and $N_{p_\mathrm{max}}=2$.
We can then simplify~\eqref{eq:PM_SRI_SCL} into three explicit formulas for the three supported G-PC cases, as shown in~\eqref{eq:PM_gpc_opt_SCL} and~\eqref{eq:Delta_PM_gpc_opt_SCL}
\begin{equation}\label{eq:PM_gpc_opt_SCL}
    \widehat{\PM}_{r}^{l,k} = \PM_s^l + \sum\limits_{j=0}^{2^s-1}| \widetilde{\psum}_{s}^{l}[j] - \widetilde{\psum}_{s}^{l,k}[j]| |\llr_{s}^{l}[j]| +  \pf{\Delta^{l,k}} \, ,
\end{equation}
where $\pf{\Delta^{l,k}}$ from~\eqref{eq:PM_gpc_opt_SCL} is expanded as
\begin{equation}\label{eq:Delta_PM_gpc_opt_SCL}
    \pf{\Delta^{l,k}} = \begin{cases}
        0&\text{if $N_p=0$} \\
        \gamma^{l,k}  |\llr_r^{l,k}[\epsilon^{l,k}]|,&\text{if $N_p=1$} \\
        \gamma_{\EV}^{l,k}|\llr_r^{l,k}[\epsilon_{\EV}^{l,k}]| + \gamma_{\OD}^{l,k}|\llr_r^{l,k}[\epsilon_{\OD}^{l,k}]|,&\text{if $N_p=2$} \\
    \end{cases}
\end{equation}
where $\gamma^{l,k}$ is the SPC node parity, $\gamma_{\EV}$ and $\gamma_{\OD}$ the parities of the even and odd nodes in the TYPE-III node, respectively, $\epsilon$ refers to the index of the minimum magnitude LLRs of the SPC node, and $\epsilon_{\EV}$ and $\epsilon_{\OD}$ the indices of the corresponding minimum magnitude LLRs of even and odd nodes in the TYPE-III node.

\begin{figure}[t]
    \centering
    \includegraphics[width=\figwidth\columnwidth]{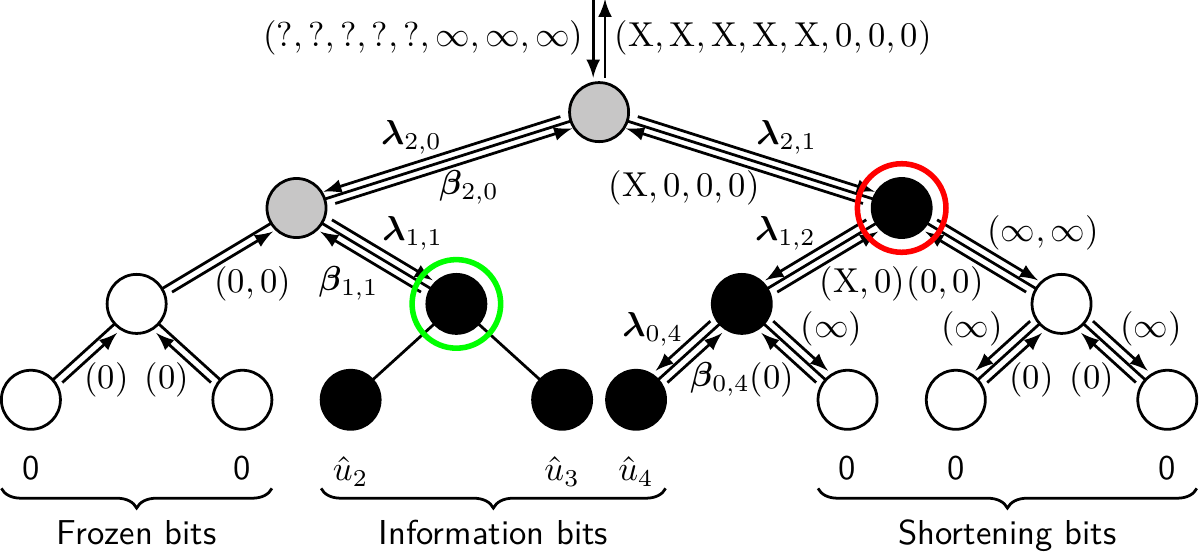}
    \caption{SC decoder tree with two frozen bits, three information bits, and three shortening bits. An arbitrary LLR value is indicated as $?$, a PSUM value $\mathrm{X}$ which can be either 0 or 1, and $\infty$ is for an LLR value for the shortening. The green and red encircled nodes are the first and last nodes, respectively.}
    \label{fig:graph_5g_rate_matching}
\end{figure}

\subsection{General Latency Optimizations}\label{sec:opt_general_latency_opt}

In this section, we describe two additional decoding optimizations: \pf{rate-matching adaptation and empirical path forking}, which can be used in any node-based SCL decoder.

\subsubsection{Rate-Matching Adaptation}\label{sec:opt_5g_rm_opt}

In 5G NR, rate-matching is applied to a codeword using either puncturing, shortening, or bit-repetition~\cite{egilmez2019development}.
For puncturing, the $N-E$ first codeword bits are removed and the LLRs after rate-recovery have value zero.
With shortening, the $N-E$ last codeword bits are removed, and the LLRs are set to $\infty$ after rate-recovery.
In our decoder, we use puncturing and shortening bits along with the first information bit index to speed up the decoding.

In Fig.~\ref{fig:graph_5g_rate_matching} we show an example of using the first information bit and the shortening bits in the decoding.
As the PSUMs at frozen bits are always 0, all R0 nodes have an output PSUM vector of all 0s and the first R0 nodes can thus be skipped.
Decoding can start directly at the green encircled node.
For the shortening bits, we treat a node with only information and shortening bits as a R1 node, like the red encircled node in Fig.~\ref{fig:graph_5g_rate_matching}.
The decoder can immediately decode this node instead of first traversing the corresponding subtree.
We note that this method does not affect the error-correcting performance and it is cheap to implement in hardware.
Note that if puncturing is used instead of shortening, the puncturing bits can be combined together with the frozen bits at the front to potentially skip even more R0 nodes.

\subsubsection{Empirical Path Forking}\label{sec:opt_empirical_path_forking}

All aforementioned algorithms for SR-List decoding work without any error-correcting performance loss.
However, some of the $\min(L-1, K_{s})$ path forks are often redundant for high-rate nodes (i.e., R1, SPC, and TYPE-III \pf{nodes}), where the bit-channels have high reliability.
In practice, the number of path forks can be empirically reduced to a constant value  \pf{$T_\mathrm{SNT}$} for a general node, as first proposed in~\cite{hashemi2017fastflexible}, which only negligibly degrades the error-correcting performance and increases the throughput by reducing the latency from path forking.
With this optimization, the number of path forks changes as follows
\begin{equation}\label{eq:empirical_path}
    \min(L-1,K_s)\rightarrow\min(\pf{T_\mathrm{SNT}},K_{s}) \,,
\end{equation}
where \pf{$T_\mathrm{SNT}$} in our decoder is either $T_{\mathrm{R1}}$, $T_{\mathrm{SPC}}$, or $T_{\mathrm{TYPE-III}}$, depending on the node type.

Using different values of \pf{$T_\mathrm{SNT}$}
\pf{as an upper limit on the number of path forks} provides further flexibility to trade-off error-correcting performance for the speed of SR-List decoding.
Numerical results shows that for $L=4$, we can select $T_{\mathrm{R1}}=1$, $T_{\mathrm{SPC}}=2$, and $T_{\mathrm{TYPE-III}}=2$.
For $L=8$ we can select $T_{\mathrm{R1}}=2$, $T_{\mathrm{SPC}}=3$, and $T_{\mathrm{TYPE-III}}=3$.
These values result in a negligible reduction in FER, as illustrated in Section~\ref{sec:SR_List_FER}, and provide a large latency reduction.
As each path fork generally requires one clock cycle (CC), this technique helps to significantly reduce the decoding latency.

\begin{figure}[t]
  \centering
  \begin{subfigure}[t]{\columnwidth}
    \centering
    \includegraphics[width=\figwidth\columnwidth]{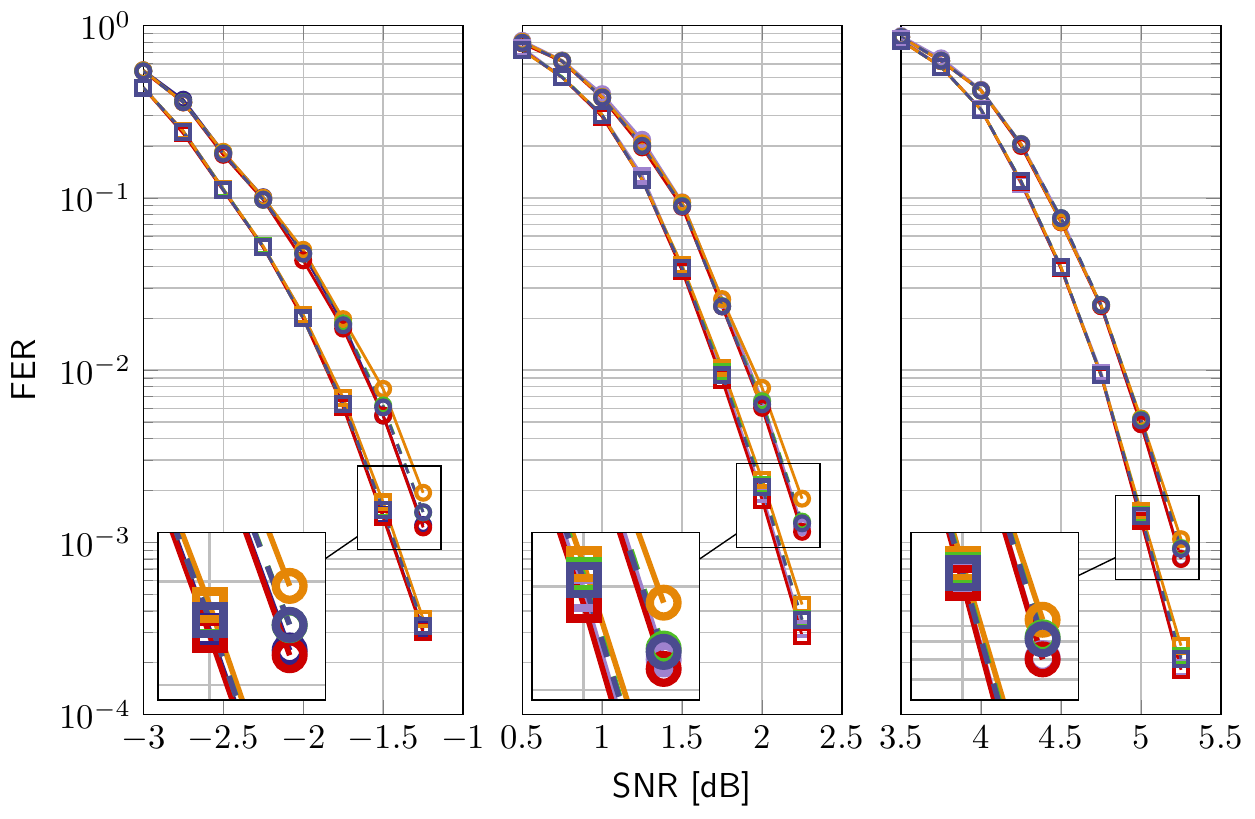}
    \caption{\pf{UL-$(1024,256)$, UL-$(1024,512)$, and UL-$(1024,768)$}}
    \label{fig:SR_UL_FER}
  \end{subfigure}
  \vspace*{0.25cm}

  \begin{subfigure}[t]{\columnwidth}
    \centering
    \includegraphics[width=\figwidth\columnwidth]{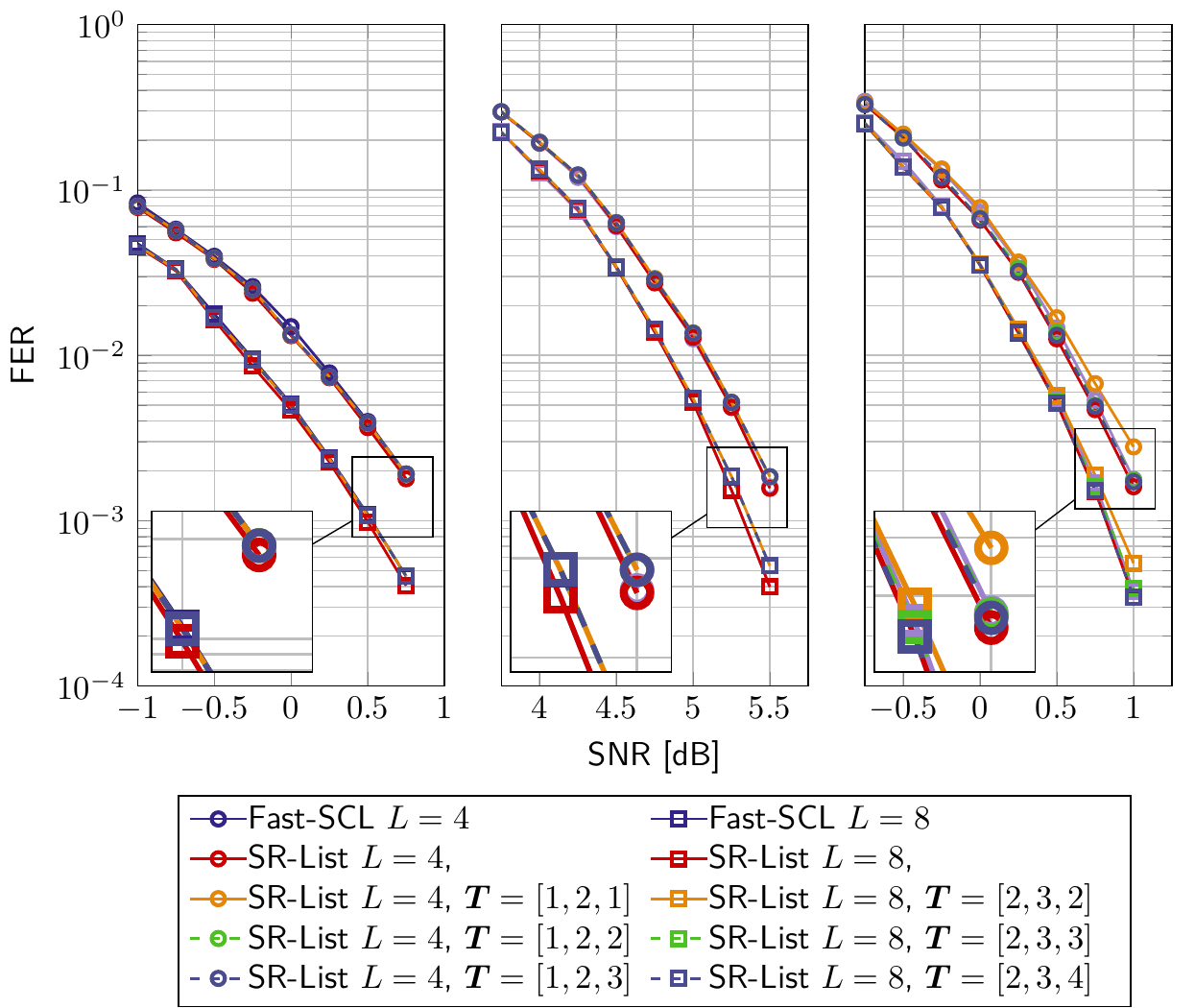}
    \caption{\pf{DL-$(108,12)$, DL-$(216,140)$, and DL-$(432,140)$}}
    \label{fig:SR_DL_FER}
  \end{subfigure}
  \caption{\pf{FER performance of Fast-SCL\cite{hashemi2017fastflexible}, SR-List, and SR-List with empirical path forking for UL and DL codes with $L \in \{4,8\}$. We use $\bm{T} = [T_\mathrm{R1}, T_\mathrm{SPC}, T_\mathrm{TYPE-III}]$ to indicate the number of path forks for each node type.}}
  \label{fig:SR_FER}
\end{figure}

\begin{figure*}[t]
    \centering
    \includegraphics[width=0.76\linewidth]{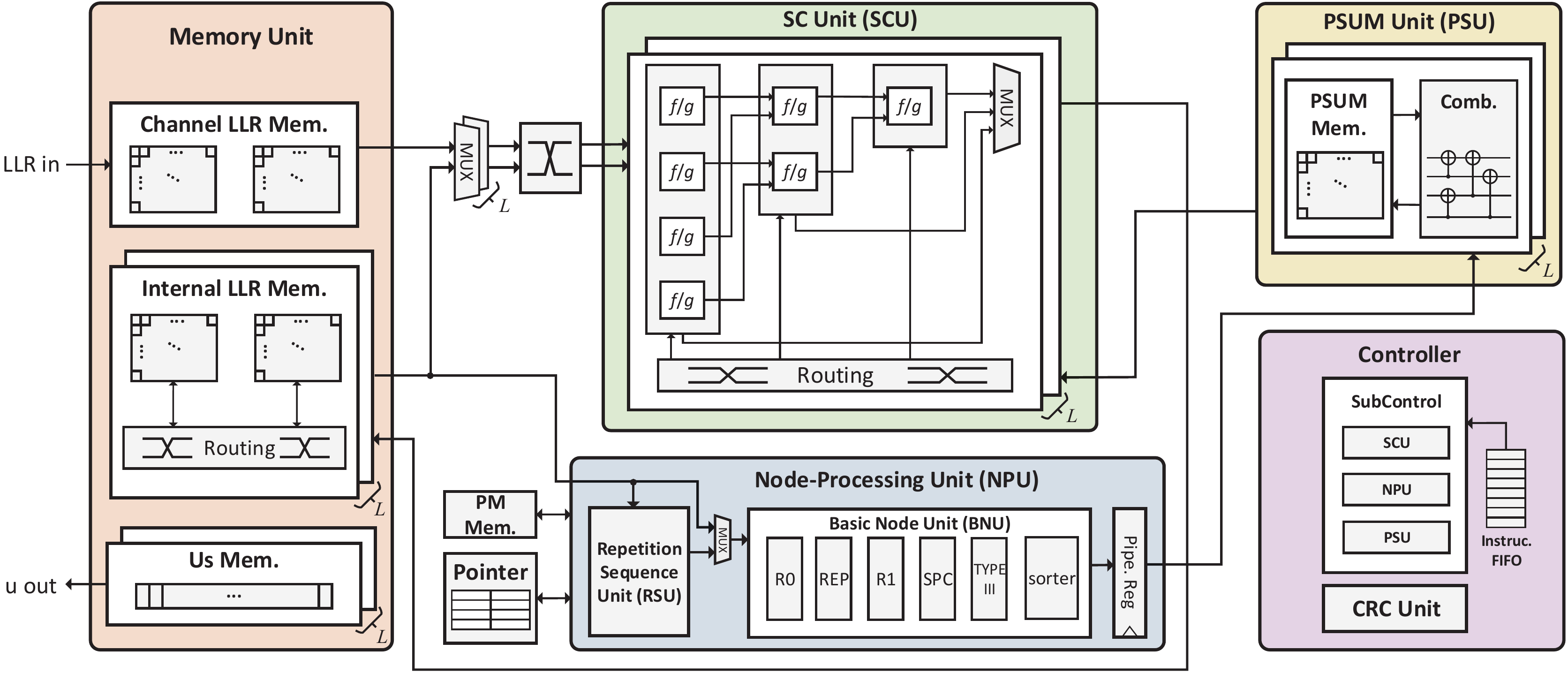}
    \caption{High-level overview of the node-based SR-List polar decoder architecture. The decoder comprises various memories, the SCU for $f/g$-operations, the NPU for node-based decoding, the PSU for PSUM combinations, and control logic.}
    \label{fig:architecture_SCL}
\end{figure*}

\subsection{Error-Correcting Performance of Proposed Algorithms}\label{sec:SR_List_FER}

In this section, the error-correcting performance of the proposed algorithms are provided by simulating 5G NR polar codes.
For UL channels (PUCCH/PUCSH), we generally only use $E=1024$ with different $R$.
For DL channels (PDCCH), we consider typical values of $E$ $\in\{108, 216, 432, 864, 1728\}$ and $A \in [12, 140]$, as described in~\cite{egilmez2019development}.
Note that most $E \in \{864, 1728\}$ bit patterns overlap with $E=432$ due to bit-repetition and we thus only use $E \in \{108,216,432\}$ for DL.
\pf{For both DL and UL channels we limit $L_{\mathrm{max}}$ to $8$ as this is what has been selected as the baseline for the maximum list size by 3GPP during the standardization process~\cite{5GSCL}}.
All data is modulated by binary phase-shift keying and sent over an additive white Gaussian noise channel.

Fig.~\ref{fig:SR_FER} illustrates the FER performance for Fast-SCL decoding~\cite{hashemi2017fastflexible}, for our SR-List decoding algorithm, and for the approximate SR-List decoding with different \pf{$T_\mathrm{SNT}$} for \pf{UL and DL codes with $L\in\{4,8\}$.}
We note that due to the selection of the globally optimal $L$ paths from the $|\bbS|L$ paths, our SR-List decoding algorithm show a slight performance improvement compared with Fast-SCL~\cite{hashemi2017fastflexible}.
For the empirical path forking, the results in Fig.~\ref{fig:SR_FER} show that for $L=4$, the constants can be selected as $T_\mathrm{R1}=1$, $T_\mathrm{SPC}=2$, and $T_\mathrm{TYPE-III}=2$ with minimal error-correcting performance degradation.
For $L=8$ we can select $T_\mathrm{R1}=2$, $T_\mathrm{SPC}=3$, and $T_\mathrm{TYPE-III}=3$.

\begin{figure*}[t]
    \centering
    \includegraphics[width=0.78\linewidth]{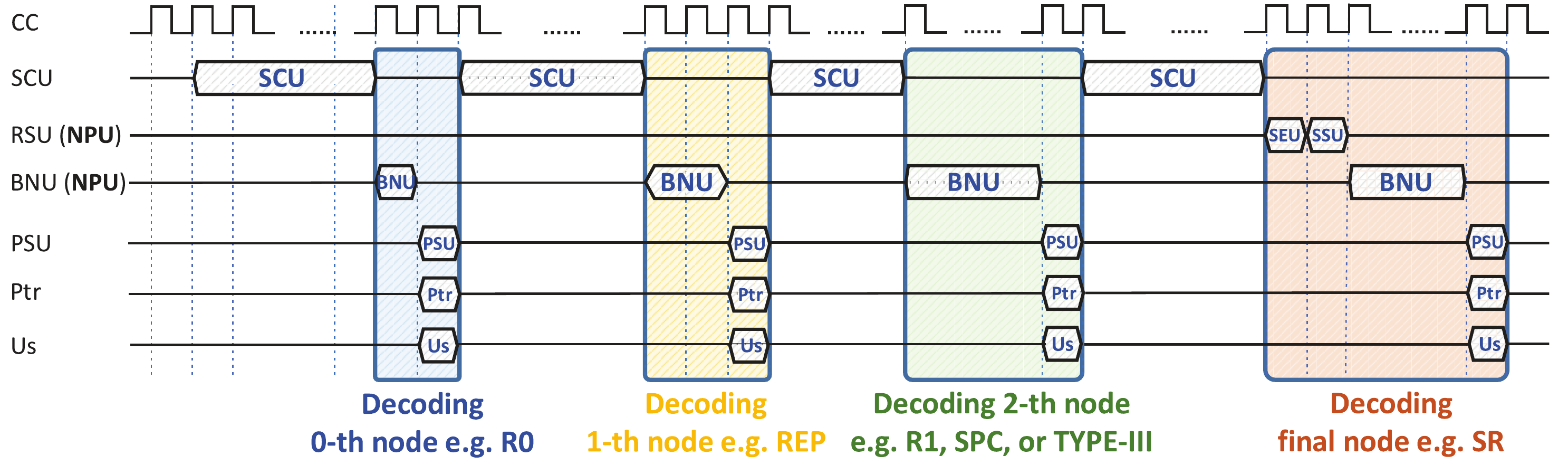}
    \caption{Example timing schedule of the proposed SR-List polar decoder.}
    \label{fig:SRL_timing}
\end{figure*}

\section{Node-Based Polar Decoder Architecture}\label{sec:hw_node_based_polar_dec_arch}

In this section, we describe the implementation of our node-based SCL decoder, which is compatible with all 5G NR polar codes and with the algorithmic optimizations described above.
More specifically, we give a detailed description of each component, which, at a high-level, comprises three types of functions: memory, decoder cores, and control logic.
High-level illustrations of the architecture and timing schedule are shown in Fig.~\ref{fig:architecture_SCL} and Fig.~\ref{fig:SRL_timing}, respectively.

For the memory, our decoder has channel LLR and internal LLR memories, Us~memory, PSUM memory, PM memory, a pointer memory, and an instruction FIFO.
The decoder cores are used to process the LLRs to get the codeword.
They comprise the SC unit (SCU), the node-processing unit (NPU), and the PSUM unit (PSU).
As illustrated in Fig.~\ref{fig:SRL_timing}, the SCU first calculates the internal LLRs until a special node is encountered.
Then, the NPU decodes the special node using the LLRs calculated by the SCU.
\pf{Note that the NPU is the generalized node implementation that supports list decoding for all SR nodes with the constraints from Section~\ref{sec:investigation_sr_gpc_5g}.}
Finally, the PSU is used for the PSUM calculations and for copying the PSUMs as required by the list decoding.
In parallel to the PSU, the decoder updates the message bits in the Us~memory and the corresponding indices in the pointer memory.
Afterwards, the SCU runs again until the next node is encountered and the procedures repeat until the final node is decoded.
This process is orchestrated by the controller, which reads from a list of instructions and coordinates the components of the decoder.

\subsection{Decoder Memories}\label{sec:hw_dec_mem}

The largest memories are the channel LLR memory and the internal LLR memory, shown on the left side of Fig.~\ref{fig:architecture_SCL}.
These store the received channel LLRs and the internal LLRs which are generated in the decoder.
Note that while the number of bits used for channel and internal LLR values is the same, the internal LLR memory is instantiated $L$ times to store the internal LLRs for the $L$ paths.
The internal LLR memory of each path is divided into two parts:
an upper part for stages closer to the root where almost no special nodes can be identified and a lower part where special nodes are frequent.
The upper part of the memory can be pruned to only store LLRs for some stages as described in Section~\ref{sec:hw_scu_flexible_multistage}, while the lower part provides storage for all remaining stages to avoid limiting the optimal use of the NPU.
The internal LLR memory makes use of another memory, the pointer memory, shown at the bottom of Fig.~\ref{fig:architecture_SCL}.
This pointer memory tracks the physical location of the internal LLRs for the individual paths in the $L$ LLR memories based on the path forking history~\cite{Alex2014TCASII}.

The decoder also contains various smaller memories.
The PSUM memory, shown in the PSU in the upper right corner of Fig.~\ref{fig:architecture_SCL}, holds the PSUM values generated after running the NPU, and is further explained in Section~\ref{sec:hw_psu}.
Note that to support the optimizations described in Section~\ref{sec:opt_general_latency_opt} the PSUM memory is always cleared to 0 before decoding a new codeword.
The Us~memory, shown in the lower left corner of Fig.~\ref{fig:architecture_SCL}, stores the decoded message bits which are generated by re-encoding the PSUMs.
At the end of the decoding process, once the CRC has been calculated for all final paths, the decoded message can be read out from the Us~memory.
The PM memory, shown above the pointer memory in Fig.~\ref{fig:architecture_SCL}, holds the PMs used during the decoding.
Finally, the instruction FIFO holds the instructions which are fed to the controller.

\begin{figure}[t]
    \centering
    \includegraphics[width=\figwidth\columnwidth]{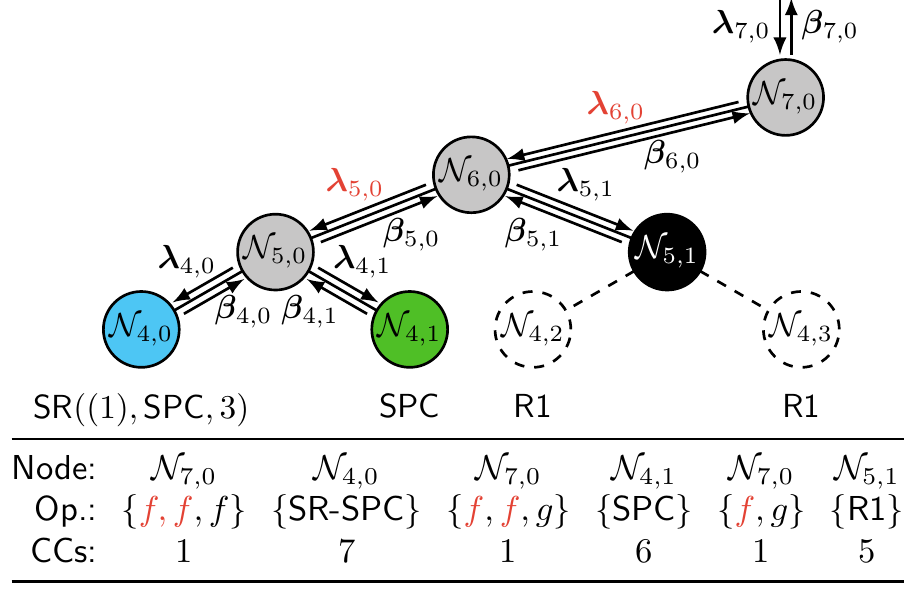}
    \caption{Flexible multi-stage decoding for $L=4$ with decoding tree and operation schedule for a 3 \pf{stage} SCU with 64 PEs in the first \pf{stage}. LLRs in red are not saved to memory. Node latencies are obtained from Table~\ref{tab:npu_cycle_cost} with 1 CC added for updating the PSUM memory using the PSU.}
    \label{fig:graph_scu}
\end{figure}

\subsection{SCU with Flexible Multi-Stage Decoding}\label{sec:hw_scu_flexible_multistage}

The SCU, shown at the top of Fig.~\ref{fig:architecture_SCL}, calculates the internal LLRs using processing-elements (PEs) that implement the $f$- and $g$-functions in~\eqref{eq:fg_func_defA} and~\eqref{eq:fg_func_defB}, respectively.
The SCU calculates the LLRs for multiple stages of the decoding tree in a single cycle by connecting PEs to process multiple \pf{stages} in a longer combinatorial path.

With multi-stage decoding, our decoder can significantly reduce the internal LLR memory size since intermediate LLRs are re-calculated  combinatorially (i.e., without the need for dedicated CCs) rather than stored.
This is especially important in list decoding, as the internal LLR memory, which is the largest of the memories, is instantiated $L$ times.
The concept of multi-stage decoding with memory reduction is shown in Fig.~\ref{fig:graph_scu}.
For an SCU with 3 \pf{stages} and $64$ PEs in the first \pf{stage}, $\llr_{7,0}$ can be read to directly calculate $\llr_{4,0}$ in one CC.
However, if the SCU has less than $64$ PEs, a cycle penalty compared to single stage decoding with no memory optimization is incurred.
This penalty arises since it takes multiple CCs to read $\llr_{7,0}$ and those multiple CCs are required repeatedly for all LLRs in Fig.~\ref{fig:graph_scu} that are directly calculated from $\llr_{7,0}$.
The SCU also calculates the LLRs $\llr_{5,0}$ and $\llr_{6,0}$ but these are not saved as the internal LLR memory in this example only stores every $3$rd stage starting from the highest stage.

In our polar decoder, the multi-stage decoding needs to be refined further since special nodes may be encountered before the last stage calculated inside the SCU.
In the example in Fig.~\ref{fig:graph_scu}, a straightforward $3$-stage SCU would calculate the LLRs up to $\calN_{4,0}$, $\calN_{4,1}$, $\calN_{4,2}$, and $\calN_{4,3}$, which are SR, SPC, and two R1 nodes, respectively, with a cost of $27$ CCs.
However, the larger R1 node $\calN_{5,1}$ could be decoded directly in only $21$ CCs (\SI{28.6}{\percent} less CCs) if the intermediate LLRs $\llr_{5,1}$ were available.
Therefore, we employ flexible multi-stage decoding, which stops at the first special node that the NPU can decode, even in the intermediate \pf{stages} of the SCU, and only saves the LLRs for this node (in the example in Fig.~\ref{fig:graph_scu} $\llr_{5,1}$).
The flexibility of returning the LLRs from intermediate \pf{stages} in the SCU increases the number of opportunities for identifying special nodes and results in an overall cycle count reduction by using the NPU earlier.
However, flexible multi-stage decoding requires that the internal LLR memory can save the LLRs for all stages where the NPU can decode nodes.
As most of the LLR memory is needed for the higher stages of the decoding tree, where node-based decoding is not used, the memory reduction is still significant, even though these lower stages can no longer be pruned.

While multi-stage decoding has been used in some previous works such as~\cite{liu20185} and~\cite{ercan2019operation}, our work improves on these in several ways.
In~\cite{liu20185}, multi-bit decoding is used at a fixed stage instead of node-based decoding which our decoder can do whenever a special node is encountered.
In~\cite{ercan2019operation}, multi-stage decoding is only used at lower stages where node-based decoding is often more efficient and their decoder can do at most two combined $f$-operations and not multiple $g$-operations.
This limitation is not present in our decoder.

\subsection{Repetition Sequence Unit (RSU) of the NPU}\label{sec:hw_rsu}

The NPU, shown at the bottom of Fig.~\ref{fig:architecture_SCL}, performs the node-based list decoding and comprises two parts based on the SR node structure described in Section~\ref{sec:opt_sr_list_dec}: the repetition sequence unit (RSU) for the SR\=/I part and the basic node unit (BNU) for the SR\=/II part.
\pf{As a generalized node implementation, the NPU can decode most SR nodes encountered in practical 5G scenarios with maximum hardware re-use to enhance the hardware efficiency.}
Note that for SR nodes where $s=r$, i.e., the SR node and the source node are the same, the SR node is decoded as a standard R0, REP, or G-PC node and the BNU is directly activated, skipping the RSU.
Similar to the special node constraints described in Section~\ref{sec:opt_sr_list_dec}, we fix the maximum node size that the NPU can decode as $N_{s_\mathrm{max}}$, which represents the trade-off between hardware complexity and decoding latency.
Specifically, increasing $N_{s_\mathrm{max}}$ allows for finding larger nodes which \pf{helps} reduce \pf{the} decoding time.
However, this also leads to a significantly more complex NPU implementation which can potentially be slower overall due to an increase in the length of the critical path.

\begin{figure}[t]
    \centering
    \includegraphics[width=\figwidth\columnwidth]{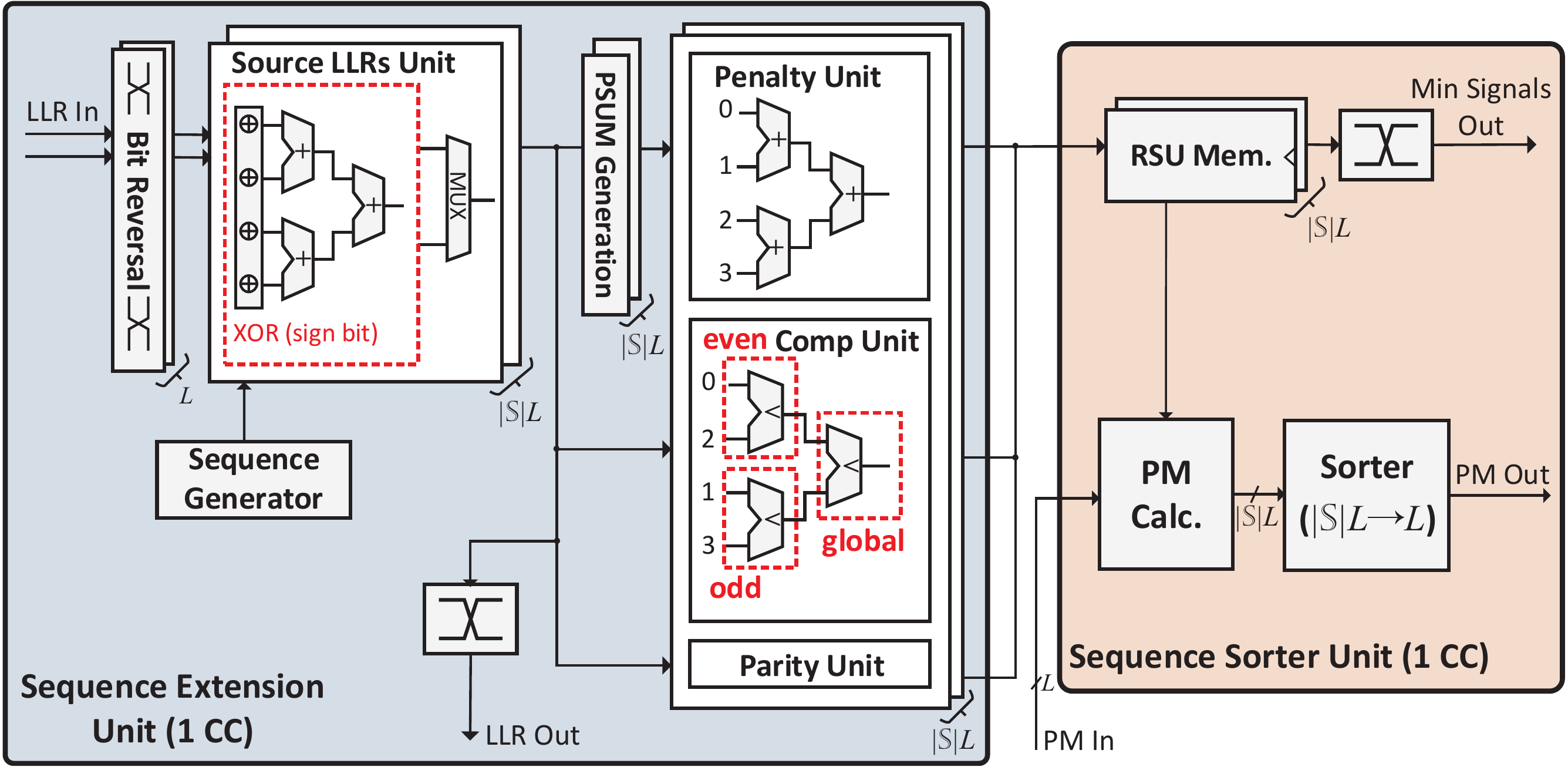}
    \caption{Repetition sequence unit (RSU) for $|\bbS|_{\max}=4$ for part 1 of the SR-List decoding in Section~\ref{sec:list_dec_r0_rep}.}
    \label{fig:sr_seq_decoder}
\end{figure}

The RSU architecture is shown in Fig.~\ref{fig:sr_seq_decoder}.
To implement~\eqref{eq:llr_source_node_SCL}, the LLR values are first put in bit-reversed order to route the LLRs correctly for different SR node sizes.
Then the repetition sequences are XOR-ed with the LLR sign-bits and the resulting LLRs are added together to generate the source node LLRs.
For adding the LLRs,  a single adder tree of depth $\log_{2}(|\bbS|_{\mathrm{max}})$ where all internal results and the final result are routed to a MUX.
As $s$ and $r$ vary, the full adder tree is not always used, and the MUX selects which internal results from the adder tree are used based on $s$ and $r$.

After calculating the source node LLRs, the penalty unit updates the PMs according to~\eqref{eq:PM_source_node_SCL} with the PSUMs from~\eqref{eq:psum_sr_est_SCL}, the comp unit finds the indices of the minimum magnitude LLRs per~\eqref{eq:gpc_index_SCL}, and the G-PC group parities are calculated based on~\eqref{eq:gpc_parity_SCL} in the parity unit.
These three steps are all done in parallel.
As the G-PC node with the constraint $N_{p_{\mathrm{max}}}=2$ is either an R1 node, an SPC node, or a TYPE-III node, we re-order the source node LLR vectors into even- and odd-indexed.
The minimum of each group and the global minimum of the node can be obtained simultaneously by using compare-and-select (CAS) trees in the comp unit.

The RSU operations are done in two CCs, with the previously described procedures implemented in the sequence extension unit (SEU) in the first CC.
In the second CC, we calculate the ML PMs according to~\eqref{eq:PM_SRI_SCL} and sort the $|\bbS|L$ PMs in the sequence sorter unit (SSU) to get the $L$ smallest ones.
The sorter architecture is further explained in Section~\ref{sec:hw_rank_sort}.

\begin{figure}[t]
    \centering
    \includegraphics[width=\figwidth\columnwidth]{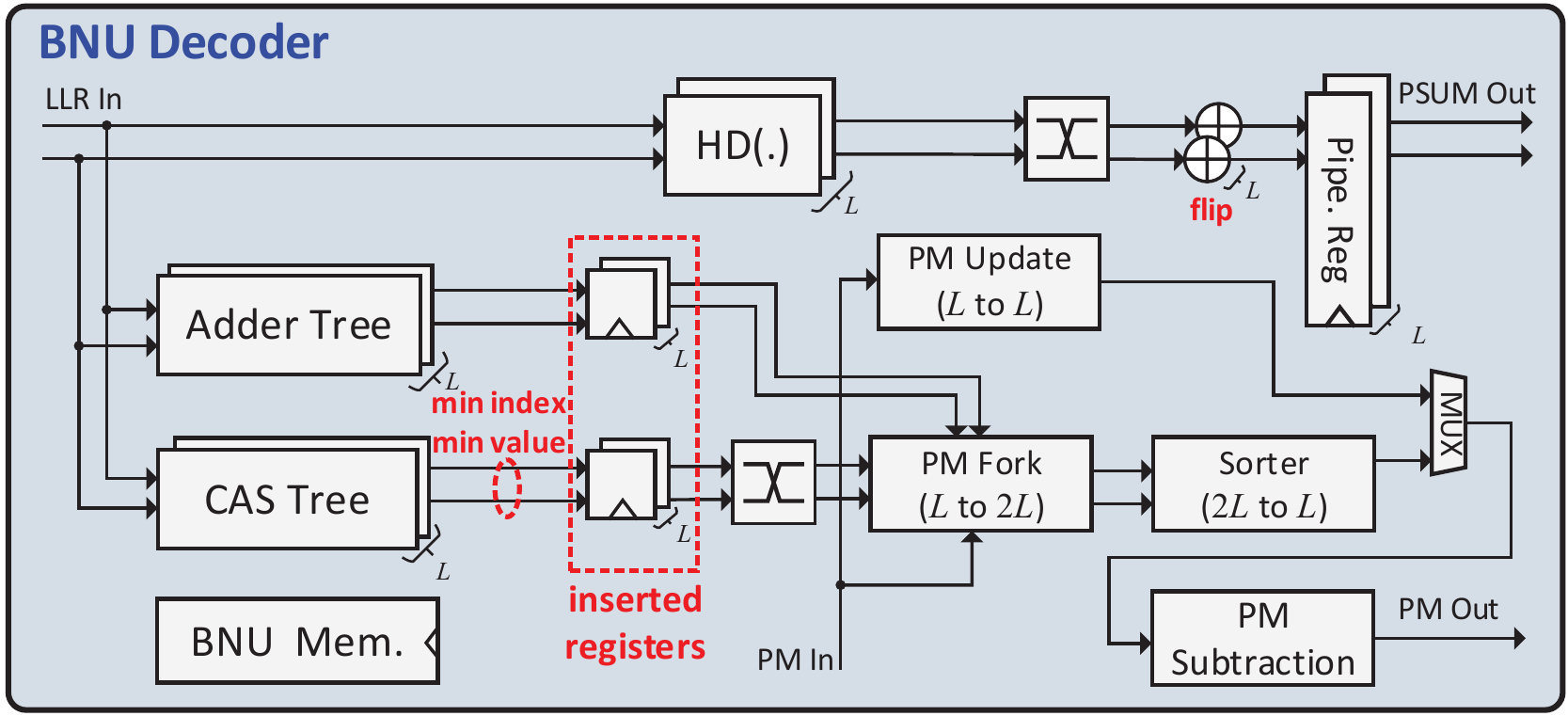}
    \caption{Basic node unit (BNU) for processing nodes such as G-PC for the SR\=/II part of SR-List decoding in Section~\ref{sec:list_dec_source_node} or directly decoding a node when $s=r$ for the SR node.}
    \label{fig:sr_npu_wrapper}
\end{figure}

\subsection{Basic Node Unit (BNU) of the NPU}\label{sec:hw_bnu}

The BNU, shown in Fig.~\ref{fig:sr_npu_wrapper}, is used for directly decoding R0 nodes, REP nodes, and G-PC nodes when $s=r$ or when these are decoded in part SR\=/II of the SR-List decoding when $s > r$.
For the low-rate nodes, i.e., R0 and REP, an adder tree calculates the PM increment for the all-zero or all-one PSUM vector from the R0 and REP nodes.
For the high rate nodes, i.e., the G-PC nodes, we instantiate a CAS tree, a PM fork unit, and a sorter for path forking.
Similar to~\cite{Hashemi2017WCNC}, the full sorter for the modified LLRs from~\eqref{eq:llr_gpc_modified_SCL} can be substituted with identifying the minimum $|\bm{\llr}_r^l[j^{l}]|$ in~\eqref{eq:PM_gpc_SCL} as only a single value is needed in every CC.
That is, the CAS tree only has to identify the $j^{l}$\=/th minimum magnitude LLR value for the $l$\=/th list in the $j$\=/th CC.
Note that for SPC and TYPE-III nodes, when $s=r$, the BNU requires an extra CC to perform Wagner decoding to ensure that the G-PC nodes satisfy the even-parity constraint, otherwise, the BNU directly executes the path forking as done for R1 nodes.
To improve the maximum clock frequency of the BNU, we insert registers after the adder and the CAS trees, which increases the latency of the BNU by one CC for path forking the G-PC nodes.
Furthermore, we subtract the minimum PM value from the final PMs as in~\cite{liu20185} at the end of the BNU to decrease the PM bit-width.

\begin{table}[t]
    \captionsetup{font=small, justification=centering}
    \centering
    \caption{\MakeUppercase{NPU decoding time for each node type depending on the list size $L$, the node size} $N_s$, \MakeUppercase{and the information bit sizes} $K_s$ \MakeUppercase{and} $K_r$.}
    \label{tab:npu_cycle_cost}

    \small
    \tabcolsep 6mm

    \begin{tabular}{ll}
      \toprule
      \multicolumn{1}{l}{Node} & \multicolumn{1}{l}{Clock cycles (CCs)} \\
      \midrule
      R0            & 1 \\
      REP           & 2 \\
      R1            & $\min(L-1,N_{s})+\redT{\bm{1}}$ \\
      SPC           & $\greenT{\bm{1}}+\min(L-1,K_{s})+\redT{\bm{1}}$ \\
      \pf{TYPE-III} & $\greenT{\bm{1}}+\min(L-1,K_{s})+\redT{\bm{1}}$ \\
      SR(R1)        & $\blueT{\bm{2}}+\min(L-1,N_{r})$ \\
      SR(SPC)       & $\blueT{\bm{2}}+\min(L-1,K_{r})+\redT{\bm{1}}$ \\
      SR(TYPE-III)  & $\blueT{\bm{2}}+\min(L-1,K_{r})+\redT{\bm{1}}$ \\
      \bottomrule
    \end{tabular}
    \begin{tablenotes}
        \footnotesize
        \item[*] $\redT{\bm{1}}$ is for the inserted registers in Fig.~\ref{fig:sr_npu_wrapper} for the BNU.
        \item[*] $\greenT{\bm{1}}$ is for the Wagner decoding for SPC and TYPE-III.
        \item[*] $\blueT{\bm{2}}$ is the fixed latency of the two parts of the RSU module in Fig.~\ref{fig:sr_seq_decoder}.
    \end{tablenotes}
\end{table}

The NPU latency for the different node types is summarized in Table~\ref{tab:npu_cycle_cost}.
Note that when decoding the SR nodes, since the indices of the minimum magnitude LLRs in~\eqref{eq:gpc_index_SCL} and the ML PMs of~\eqref{eq:PM_SRI_SCL} are calculated in the RSU, one CC for SR(R1) and one Wagner decoding CC for SR(SPC) and SR(TYPE-III) can be removed.
A more detailed decoding latency analysis is given in Section~\ref{sec:impl_results_latency_analysis}.

\subsection{PSUM Unit (PSU)}\label{sec:hw_psu}

The PSU, shown in the upper right corner of Fig.~\ref{fig:architecture_SCL}, performs the PSUM operations and it contains the PSUM memory.
When the NPU completes the processing of a node, the PSU is activated in the next CC, as shown in Fig.~\ref{fig:SRL_timing}.
The PSU then either directly updates the PSUM memory with the node output when decoding a left descendant node or performs the combine-operation in~\eqref{eq:psum_func} and then updates the PSUM memory for a right descendant node.
It takes one CC to do the combine-operation starting from any stage up to the maximum stage, which is stage 8 for DL and stage 9 for UL with 5G NR polar codes.
After a node is decoded, the PSUM memory directly copies the PSUM values based on the parent paths of the surviving paths.

\subsection{Partial Rank-Order Sorter}\label{sec:hw_rank_sort}

The critical path of an SCL decoder is generally through the PM sorting unit~\cite{Alex2015LLRbased}, and several works have thus focused on optimizing the PM sorting, \pf{either using pruned sorting~\cite{Alex2015ISCAS, Yong2016TCASII, Gal2020SiPS} or partial sorting~\cite{Fan2016JSAC, Liang2016Globcom, Chui2018hroushold}}, with a comparison of various implementations in~\cite{Gal2020SiPS}.
\pf{Most of these current sorters focus on the $2L\rightarrow L$ sorting required for the result in~\eqref{eq:PM_SCL_func}, including the SOA pruned rank-order sorter in~\cite{Gal2020SiPS}, which prunes the sorter by using the property that the result of comparing certain PMs is known ahead of time.}
However, this property of the PMs depends on how the path forking is done in hardware and requires additional routing and control.
\pf{As the SR-I part requires sorting $|\mathbb{S}|L$ PMs from~\eqref{eq:PM_SRI_SCL}, the sorting size for the SR-List decoder is much larger compared to previous works, and all of the $|\mathbb{S}|L$ PMs have different PM penalties so there is no explicit order known ahead of time.}
Therefore, in our decoder, we consider more general cases and assume no special properties on the input PMs.

\pf{To avoid a large complexity by fully sorting all input PMs,} we simplify the sorter by allowing the \pf{output} PMs to be partially sorted, as SCL decoding only requires the $L$ smallest PMs \pf{without these being sorted}.
\pf{Herein, a full sorter can be replaced by two small full sorters and one layer of comparators, which results in our proposed partial rank-order sorter almost halving the number of comparators compared to a full sorter.}

Let $X$-to-$Y$ denote a sorter with $X$ inputs and $Y$ outputs.
Our partial rank-order sorter, shown in Fig.~\ref{fig:rank_partial_sort_block} for $8$-to-$4$, is composed of two $\frac{X}{2}$-to-$Y$ full rank-order sorters which return at least $Y$ fully sorted results in ascending and descending order, respectively, as shown in~\eqref{eq:DRS_m} and~\eqref{eq:DRS_n}
\begin{alignat}{3}
   &\text{Full rank (ascending):} \quad&&m_{0} \leq m_{1} \leq \cdots \leq m_{Y-1}, \label{eq:DRS_m} \\
   &\text{Full rank (descending):} \quad&&n_{0} \geq n_{1} \geq \cdots \geq n_{Y-1}. \label{eq:DRS_n}
\end{alignat}
After the full rank sorters, we use a set of comparators to find the $Y$ smallest inputs from the pairs $\{m_{i},n_{i}\},i\in[0,Y-1]$.
\begin{Proposition}
For an $X$-to-$Y$ partial rank sorter, all results belong to the $Y$ smallest elements from the $X$ inputs.
\end{Proposition}

\begin{proof}
Any partial rank-order sorter output ${out}_{i}$, for $i\in[0,Y-1]$, is chosen by comparing $m_{i}$ and $n_{i}$, where $m_{i}$~is less than or equal to $\frac{X}{2}-i-1$ other $m$ values and $n_{i}$ is less than or equal to $\frac{X}{2}-Y+i$ other $n$ values.
Therefore, for every comparison between the two inputs $m_{i}$ and $n_{i}$, it holds that ${out}_{i}$ must be less than or equal to $(\frac{X}{2}-i-1) + (\frac{X}{2}-Y+i) + 1 = X-Y$ other input values and ${out}_{i}$ is therefore always amongst the $X-Y$ smallest inputs.
\end{proof}

Compared to the $\frac{1}{2}(X^2 - X)$ comparators in the full rank-order sorter, the partial rank sorter only has $(\frac{1}{4}X^2 - \frac{1}{2}X) + Y$ comparators, which simplifies to $\frac{X^2}{4}$ when $Y = \frac{X}{2}$.
From the synthesis results, shown in Table~\ref{tab:synth_distr_sort}, the proposed sorter has a \SI{45}{\percent} smaller area than the full sorter for $64$-to-$8$.

\begin{table}[t]
	\captionsetup{font=small,justification=centering}
	\centering
	\caption{\MakeUppercase{Synthesis results for different rank sorters in \SI{28}{\nano\meter} FD-SOI with a \SI{0.5}{\nano\second} target.}}
	\label{tab:synth_distr_sort}

  \small
	\tabcolsep 0.7mm
	\def\CmidW{0.08cm}

	\begin{tabular}{lcrrrcrrr}
		\toprule
		Type & ~ & \mc{3}{c}{\textbf{Partial Rank Sorter}} & ~ & \mc{3}{c}{Full Rank Sorter~\cite{Gal2020SiPS}$^{\dagger}$} \\
		\cmidrule(r{\CmidW}){1-1} \cmidrule(l{\CmidW}r{\CmidW}){2-5} \cmidrule(l{\CmidW}){6-9}
		Size                                & ~ & $16$-to-$8$ & $32$-to-$8$ & $64$-to-$8$  & ~ & $16$-to-$8$ & $32$-to-$8$ & $64$-to-$8$ \\
		\# Comp.                            & ~ & $64$        & $248$       & $1000$       & ~ & $120$       & $496$       & $2016$      \\
		Area (\SI{}{\micro\meter\squared})  & ~ & $2482$      & $9553$      & $31616$      & ~ & $3959$      & $14051$     & $57453$     \\
		\bottomrule
	\end{tabular}
	\begin{tablenotes}
		\footnotesize
		\item[*] $\dagger$ The full rank sorter from~\cite{Gal2020SiPS} was re-implemented and synthesized in \SI{28}{\nano\meter} FD-SOI.
	\end{tablenotes}
\end{table}

\begin{figure}[t]
	\centering
	\includegraphics[width=0.55\columnwidth]{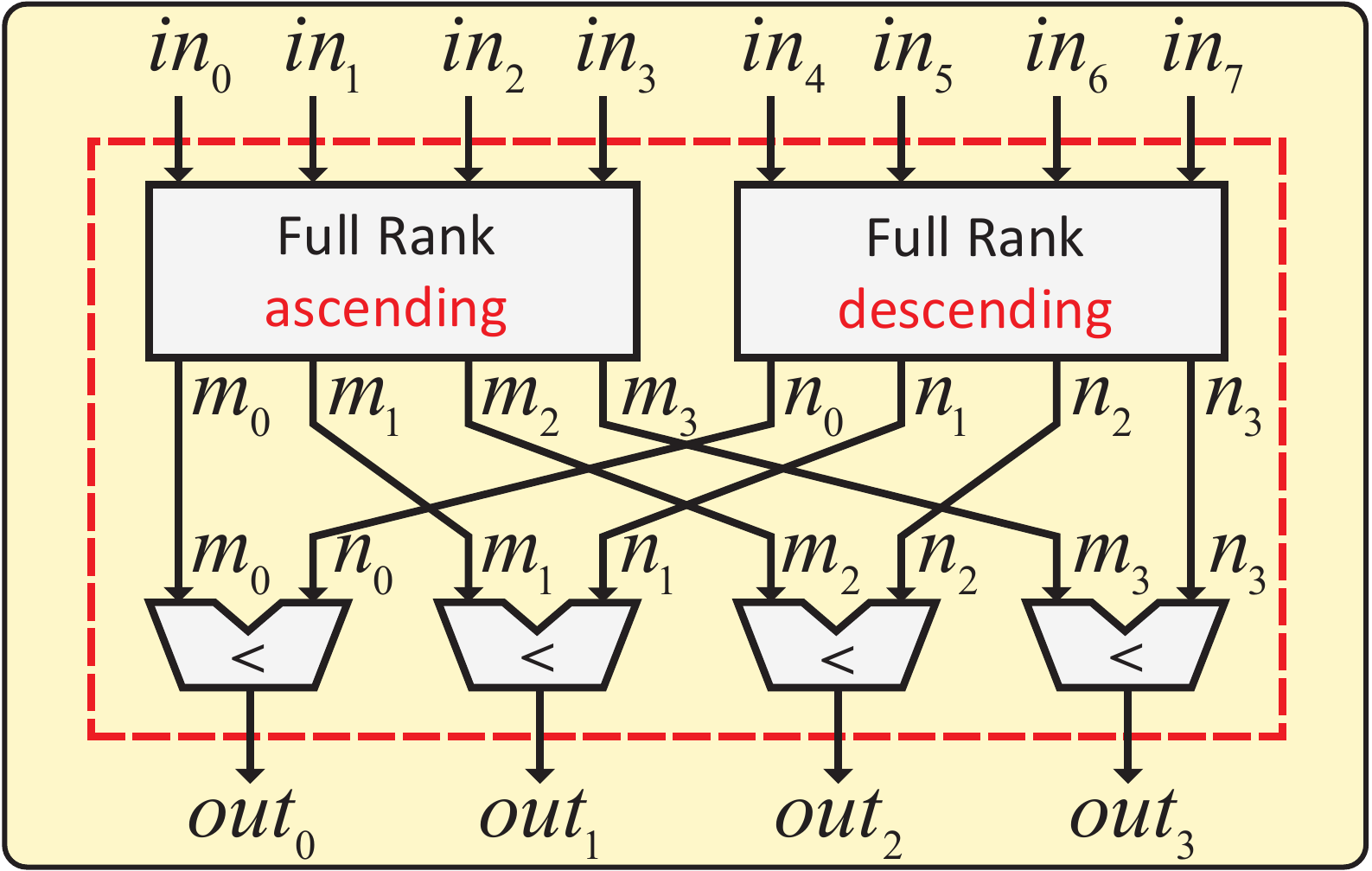}
	\caption{Architecture of the $8$-to-$4$ partial rank-order sorter.}
	\label{fig:rank_partial_sort_block}
\end{figure}

\subsection{Controller}\label{sec:hw_ctrl}

Finally, the controller, shown in the lower right corner of Fig.~\ref{fig:architecture_SCL}, is responsible for orchestrating the decoding process, by activating the different main components of the decoder, i.e., the SCU, NPU, PSU, and LLR memories.
The controller is given a list of instructions which can be generated offline or using a hardware scheme similar to the one in~\cite{hashemi2019rate}.
These instructions provide the controller with information about the current operation, e.g., the size and type of a special node for the NPU to decode.
When the controller activates a component like the NPU, the NPU sub-controller uses the high-level instruction for the actual node decoding to control the RSU, the BNU, and the path forking process.
Note that once the $g$-function has been executed at the root node of the polar decoding tree, the channel LLR memory can be overwritten with the channel LLRs of a new codeword while the controller continues to decode the current codeword to reduce the time between decoding multiple codewords.

\begin{figure}[t]
  \centering
  \begin{subfigure}[t]{\columnwidth}
    \centering
    \includegraphics[width=\figwidth\columnwidth]{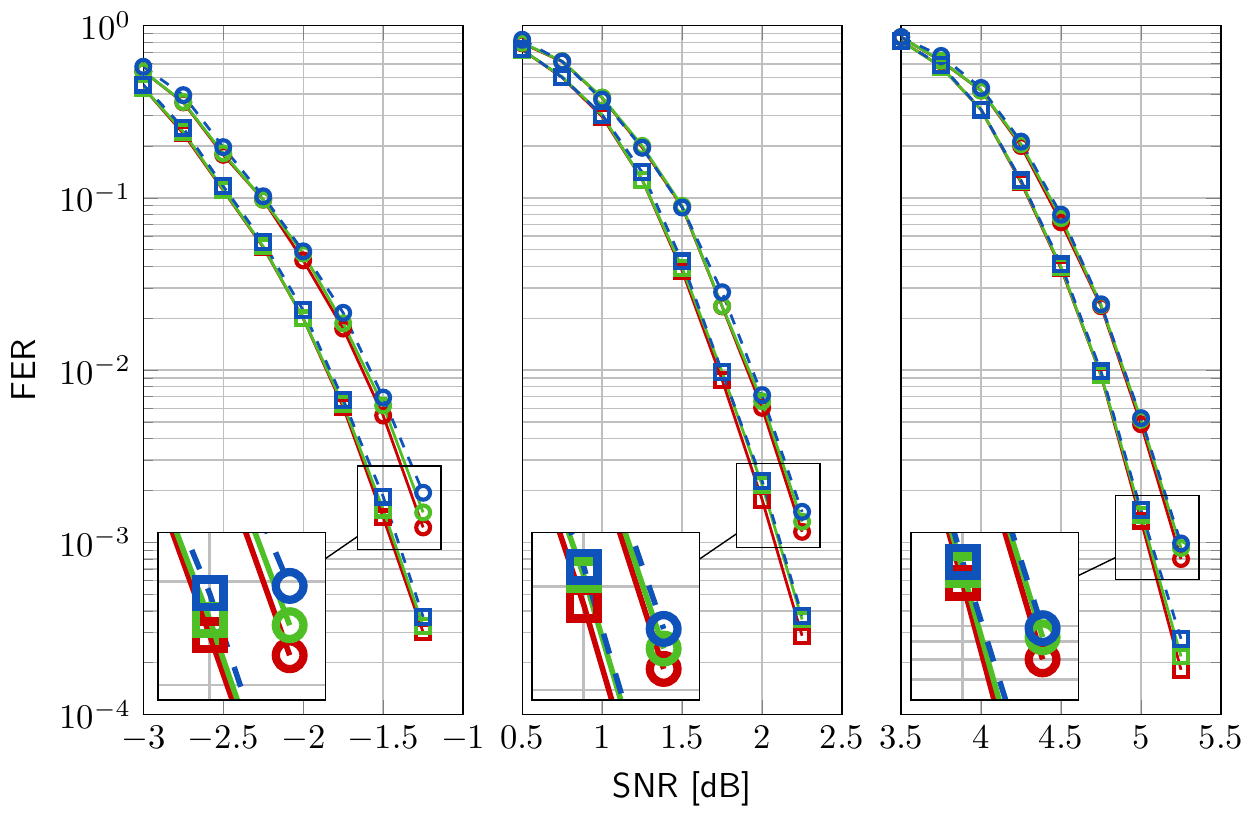}
    \caption{\pf{UL-$(1024,256)$, UL-$(1024,512)$, and UL-$(1024,768)$}}
    \label{fig:SR_UL_FER_quan}
  \end{subfigure}
  \vspace*{0.25cm}

  \begin{subfigure}[t]{\columnwidth}
    \centering
    \includegraphics[width=\figwidth\columnwidth]{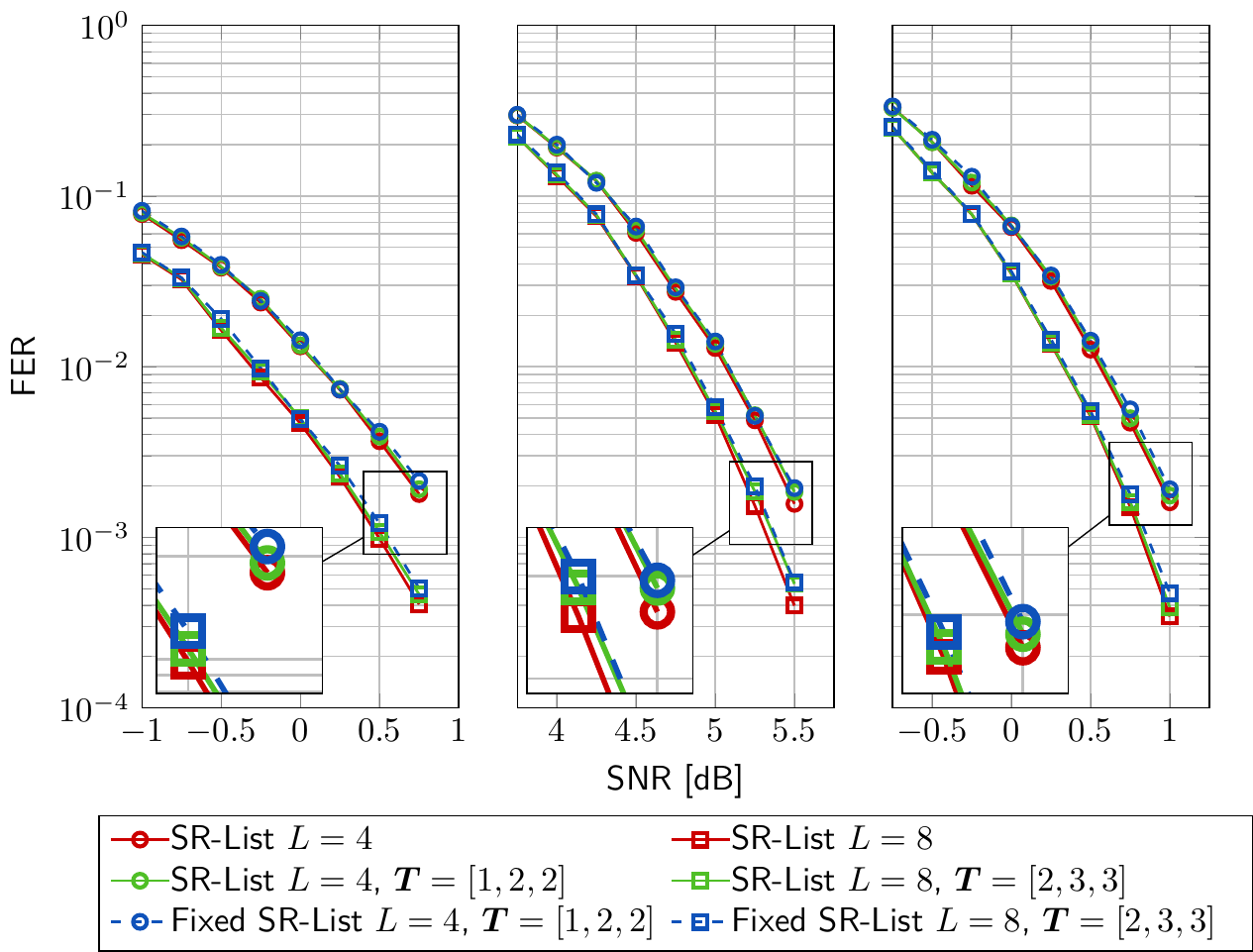}
    \caption{\pf{DL-$(108,12)$, DL-$(216,140)$, and DL-$(432,140)$}}
    \label{fig:SR_DL_FER_quan}
  \end{subfigure}
  \caption{\pf{FER comparison between floating-point and fixed-point of the proposed SR-List decoder for UL and DL polar codes with $L \in \{4,8\}$. We use $\bm{T} = [T_\mathrm{R1}, T_\mathrm{SPC}, T_\mathrm{TYPE-III}]$ to indicate the number of path forks for each node type.}}
  \label{fig:SR_FER_quan}
\end{figure}

\section{Implementation Results}\label{sec:impl_results}

In this section, we present the synthesis results for our SR-List decoder for 5G NR polar codes.
All synthesis results are based on a STM \SI{28}{\nano\meter} FD-SOI technology in the slow-slow corner and timing constraints that are not achievable to better determine the maximum achievable operating frequency for each design.
To simplify the design space, we fix the default number of SCU stages ($\#$SCU) to $2$, the number of PEs at the first \pf{stage} of the SCU ($\#$PE) to $64$, and $N_{s_\mathrm{max}}=32$.
\pf{Since the maximum polar code length for the 5G scenarios is only $N=1024$, there is a limit to the benefit of additional SCU stages given that the largest reductions occur at the higher stages and the LLR memory rows used by the NPU cannot be pruned.}
Moreover, we use the constraints discussed in Section~\ref{sec:opt_sr_list_dec} (i.e., $|\bbS|_{\max}\in\{2,4,8\}$ and G-PC nodes with $N_{p_{\max}}=2$).
Note that we present two versions of our decoder, one tailored to DL with maximum code length $512$ bits and one for UL with a maximum code length of $1024$ bits.

\subsection{Quantized FER Performance}\label{sec:impl_results_quant_fer}

In Fig.~\ref{fig:SR_FER_quan}, we show the FER performance of the decoder using floating-point and fixed-point (sign-magnitude) for \pf{UL and DL codes with $L\in\{4,8\}$}.
Let $Q_{q_{i}.q_{f}}$ denote a fixed-point number with one sign-bit, $q_{i}-q_{f}-1$ integer bits, and $q_{f}$ fractional bits.
We use $Q_{6.2}$ for the LLRs and $Q_{7.0}$ for the PMs as the quantized FER results in Fig.~\ref{fig:SR_FER_quan}.
Numerical results show that this representation has a negligible FER performance loss compared to floating point.

\begin{table}[t]
    \captionsetup{font=small, justification=centering}
    \centering
    \caption{\MakeUppercase{Implementation results for 512-bit DL polar codes.}}
    \label{tab:hw_dl_results}

    \small
    \tabcolsep 0.7mm
    \def\CmidW{0.08cm}

    \resizebox{\columnwidth}{!}{\begin{tabular}{@{}l rrr rrr rrr@{}}
        \toprule
        ~ & \mc{3}{c}{$L=2$} & \mc{3}{c}{$L=4$} & \mc{3}{c}{$L=8$}   \\
        $|\mathbb{S}|_{\max}$ & \mc{1}{c}{$2$} & \mc{1}{c}{$\bm{4}$} & \mc{1}{c}{$8$} & \mc{1}{c}{$2$} & \mc{1}{c}{$\bm{4}$} & \mc{1}{c}{$8$} & \mc{1}{c}{$2$} & \mc{1}{c}{$\bm{4}$} & \mc{1}{c}{$8$} \\
        \cmidrule(l{\CmidW}){2-4} \cmidrule(l{\CmidW}r{\CmidW}){5-7} \cmidrule(l{\CmidW}){8-10}
        CCs                                 & $155$   & $\bm{140}$    & {$125$}   & $170$   & $\bm{155}$    & {$140$}   & $188$   & $\bm{173}$   & $158$   \\
        Latency [\SI{}{\micro\second}]      & $0.109$  & $\bm{0.101}$   & {$0.112$}  & $0.130$  & $\bm{0.119}$   & {$0.130$}  & $0.185$  & $\bm{0.174}$  & $0.167$  \\
        Area [\SI{}{\milli\meter\squared}]  & $0.099$ & $\bm{0.109}$  & {$0.135$} & $0.195$ & $\bm{0.211}$  & {$0.270$} & $0.405$ & $\bm{0.439}$ & $0.565$ \\
        Freq. [\SI{}{\mega\hertz}]          & $1418$  & $\bm{1385}$   & {$1120$}  & $1312$  & $\bm{1302}$   & {$1081$}  & $1016$  & $\bm{994}$   & $944$   \\
        T/P [\SI{}{Gbps}]                   & $4.685$ & $\bm{5.065}$  & {$4.587$} & $3.952$ & $\bm{4.301}$  & {$3.954$} & $2.768$ & $\bm{2.942}$ & $3.060$ \\
        \cmidrule(l{\CmidW}){2-4} \cmidrule(l{\CmidW}r{\CmidW}){5-7} \cmidrule(l{\CmidW}){8-10}
        \parbox{1.8cm}{\raggedright Area Eff.\\ {[\SI{}{Gbps\per\milli\meter\squared}]} } & $47.09$ & $\bm{46.68}$ & {$33.99$} & $20.31$ & $\bm{20.43}$ & {$14.62$} & $6.834$ & $\bm{6.704}$ & $5.420$ \\
        \bottomrule
    \end{tabular}}
    \begin{tablenotes}
        \footnotesize
        \item[*] While our decoder is compatible with all DL codes, the results shown in this table are for the worst-case latency for all DL codes.
    \end{tablenotes}
\end{table}

\subsection{Design-Space Exploration}\label{sec:impl_results_design_exploration}

While we already fixed some parts of the decoder architecture by setting $\#\text{SCU}=2$, $\#\text{PE}=64$, $N_{s_\mathrm{max}}=32$ and restricting the G-PC node to $N_p \in \{0, 1, 2\}$ in Section~\ref{sec:opt_sr_list_dec}, we still have to determine $|\mathbb{S}|_{\max}$.
To determine an optimal $|\mathbb{S}|_{\max}$, we perform a design-space exploration with $|\mathbb{S}|_{\max} \in\{2, 4, 8\}$ as defined in Section~\ref{sec:opt_sr_list_dec} for DL polar codes with $L\in\{2, 4, 8\}$.
The corresponding results are given in Table~\ref{tab:hw_dl_results}.
Note that our decoder is compatible with all PDCCH codes and only the worst-case hardware latency and throughput results amongst all PDCCH codes are shown in Table~\ref{tab:hw_dl_results}.

For the area in Table~\ref{tab:hw_dl_results}, we note that increasing $|\bbS|_{\max}$ from 4 to 8 results in an increase in the area of \SI{23.9}{\percent}, \SI{28.0}{\percent}, and \SI{28.7}{\percent} for list sizes 2, 4, and 8, respectively.
However, while increasing $|\bbS|_{\max}$ reduces the number of decoding cycles, once the reduction in the maximum clock frequency has been accounted for, the actual decoding latency in \SI{}{\micro\second} increases slightly when $|\bbS|_{\max}>4$ for $L=2$ and $L=4$ and there is only a small \SI{4}{\percent} reduction in the latency for $L=8$.
While $|\bbS|_{\max}=2$ has the best area efficiency for $L \in \{2, 8\}$, $|\bbS|_{\max}=4$ always has a higher throughput and the difference between $|\bbS|_{\max}=2$ and $|\bbS|_{\max}=4$ in terms of area efficiency is minimal.
Therefore, for our SR-List decoder, we set $|\bbS|_{\max}=4$ as this represents the best trade-off when decoding time has a slightly higher priority than area.

\subsection{Decoding Latency Analysis}\label{sec:impl_results_latency_analysis}

To evaluate the decoding latency reduction from the various algorithms presented in Section~\ref{sec:proposed_algos}, we calculate the number of cycles to decode each DL polar code with $E=432$ and different $A$.
The corresponding worst-case latency for PDCCH with $L=8$ is shown in Fig.~\ref{fig:latency_analysis_E432_L8}.

Compared with the baseline Fast-SCL decoder~\cite{hashemi2017fastflexible}, the proposed flexible multi-stage decoding reduces the worst-case latency from \SI{334}{CCs} to \SI{311}{CCs}.
Using our proposed optimized SR-List decoding algorithm further reduces the worst-case latency to \SI{240}{CCs}, showing that the SR node provides a significant reduction in decoding latency.
Moreover, using empirical path forking and rate-matching adaptation, our final SR-List decoder with all optimizations reduces the worst-case latency to \SI{173}{CCs} which is a \SI{48.2}{\percent} reduction in the worst-case latency compared to the Fast-SCL decoder.

\begin{figure}[t]
	\centering
	\includegraphics[width=\figwidth\columnwidth]{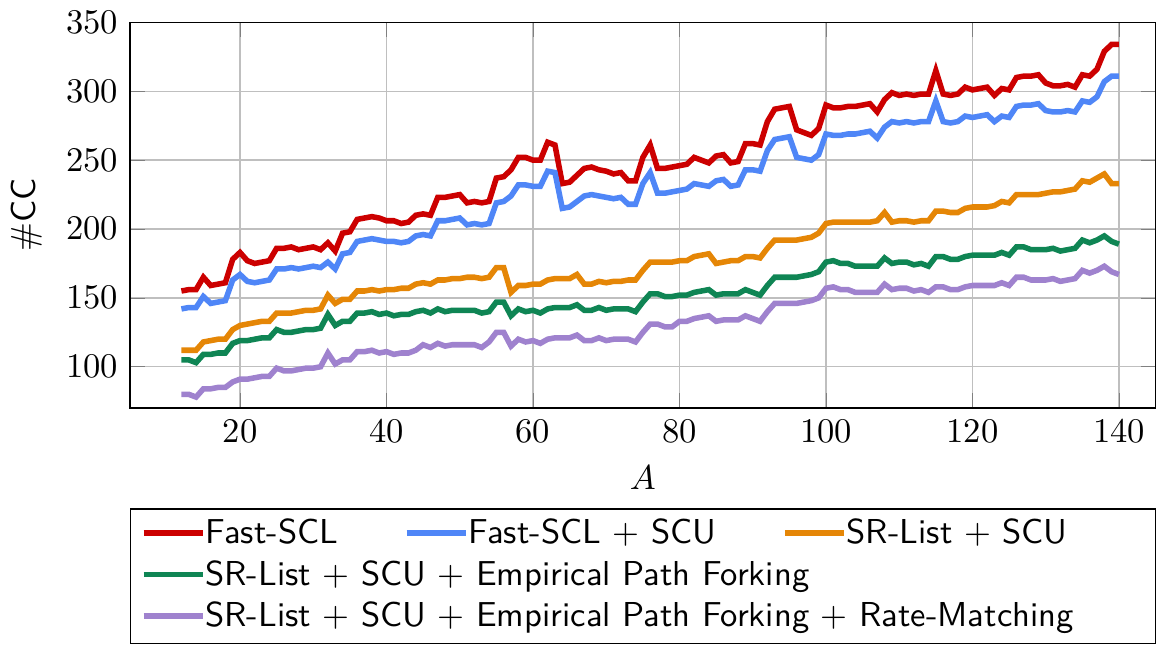}
	\caption{Latency analysis of the proposed SR-List decoder with $|\bbS|_{\max}=4$ for DL polar codes with $E=432$ and $L=8$.}
	\label{fig:latency_analysis_E432_L8}
\end{figure}

\subsection{Comparison With Previous Works}\label{sec:impl_results_comparison}

In Table~\ref{tab:hw_comparison}, we present our final implementations results for our proposed SR-List decoder for UL and compare them with the SOA architectures in~\cite{Alex2015LLRbased, giard2017polarbear, hashemi2017fastflexible, Kim2018TSP, liu20185, Lee2020TSP, tao2020configurable}.
For the comparison, we provide the results for the UL-$(1024,512)$ polar codes.
As these previous works are implemented in different technologies, we scale their results to \SI{28}{\nano\meter}.
Note that we only consider the worst-case latency for~\cite{Kim2018TSP, Lee2020TSP}.
While they include some dynamic optimization strategies, we assume that a decoder in a 5G NR modem needs to provide reliable fixed latency guarantees.
We note that in length-$1024$ UL polar codes, length-$32$ SR nodes with $s > r$ are seldom found for middle and high-rate codes,
Therefore, the RSU size is reduced to $16$ in Table~\ref{tab:hw_comparison} to improve the hardware results.

For $L=2$, our decoder has a throughput of \SI{4.457}{Gbps} and an area efficiency of \SI{29335}{Mbps\per\milli\meter\squared}, which is a $1.97\times$ improvement in area efficiency compared to~\cite{tao2020configurable} with a small \SI{4}{\percent} reduction in throughput.
For $L=4$, compared with \pf{the SCL architecture with an early stopping criterion}~\cite{Kim2018TSP}, our proposed SR-List decoder has a \SI{12.3}{\percent} improvement in throughput, but our area efficiency is \SI{33.6}{\percent} lower.
For $L=8$, the scaled throughput our decoder is \SI{2.532}{Gbps} which is $3.1\times$ higher than~\cite{liu20185} and our area efficiency is \SI{4165}{Mbps\per\milli\meter\squared} which is a $10.6\%$ improvement over~\cite{hashemi2017fastflexible}.

\begin{table*}[]
    \centering
    \captionsetup{font=small, justification=centering}
    \caption{\uppercase{Comparisons with SOA decoders for $N=1024$ and worst-case latency calculated from the UL-$(1024,512)$ polar code.}}
    \label{tab:hw_comparison}

    \tabcolsep 0.8mm

    \def\PboxW{1.5cm}
    \def\CmidW{0.08cm}
    \begin{tabular}{l ccc c c c c c c c}
        \toprule
        {~} & \mc{3}{c}{\textbf{Our work for UL$^\dagger$}} & {\parbox{\PboxW}{\centering TSP 15 \\ \cite{Alex2015LLRbased}}} & \parbox{\PboxW}{\centering JETCAS 17 \\ \cite{giard2017polarbear}} & \parbox{\PboxW}{\centering TSP 17 \\ \cite{hashemi2017fastflexible}} &  \parbox{\PboxW}{\centering TSP 18 \\ \cite{Kim2018TSP}$^*$} &  \parbox{\PboxW}{\centering ISWCS 18 \cite{liu20185}} & \parbox{\PboxW}{\centering TSP 20 \\ \cite{Lee2020TSP}$^*$} &  \parbox{\PboxW}{\centering JSSC 20 \\ \cite{tao2020configurable}} \\
\cmidrule(l{\CmidW}r{\CmidW}){2-4} \cmidrule(l{\CmidW}r{\CmidW}){5-5} \cmidrule(l{\CmidW}r{\CmidW}){6-6} \cmidrule(l{\CmidW}r{\CmidW}){7-7} \cmidrule(l{\CmidW}r{\CmidW}){8-8} \cmidrule(l{\CmidW}r{\CmidW}){9-9} \cmidrule(l{\CmidW}r{\CmidW}){10-10} \cmidrule(l{\CmidW}){11-11}

Algorithm                           & \mc{3}{c}{SR-List} & \mc{1}{c}{SCL} & \mc{1}{c}{SCL} & \mc{1}{c}{Fast-SCL} & \mc{1}{c}{Fast-SCL} & \mc{1}{c}{SCL Multi-bit} & \mc{1}{c}{Fast-SCL Flip} & \mc{1}{c}{Split-tree SCL} \\

Technology [\SI{}{\nano\meter}]     & \mc{3}{c}{ss$28$} & \mc{1}{c}{$90$}  & \mc{1}{c}{$28$} & \mc{1}{c}{$65$} & \mc{1}{c}{$65$} & \mc{1}{c}{$16$}  & \mc{1}{c}{$90$} & \mc{1}{c}{$40$} \\

List-size                           & \mc{1}{c}{$2$} & \mc{1}{c}{$4$} & \mc{1}{c}{$8$} & \mc{1}{c}{$8$} & \mc{1}{c}{$4$} & \mc{1}{c}{$8$} & \mc{1}{c}{$4$} & \mc{1}{c}{$8$} & \mc{1}{c}{$8$} & \mc{1}{c}{$2$} \\

    CCs                             & \hphantom{00}{$325$}      & \hphantom{00}{$355$}      & \hphantom{0.}{$395$}  & \hphantom{.}{$2662$}  & \hphantom{.}{$2408$}  & \hphantom{0.}{$618$}  & \hphantom{00}{$395$}  & \hphantom{0.}{$790$}  & \hphantom{0.}{$478$}  & \hphantom{00}$136$ \\

Latency [\SI{}{\micro\second}]      & \hphantom{.0}{$0.23$}     & \hphantom{.0}{$0.28$}     & \hphantom{0}$0.40$    & \hphantom{0}{$4.18$}  & \hphantom{0}{$3.34$}  & \hphantom{0}{$0.85$}  & \hphantom{0.}{$0.64$} & \hphantom{0}{$0.72$}  & \hphantom{0}{$0.80$}  & \hphantom{0.}$0.32$  \\

Area [\SI{}{\milli\meter\squared}]  & \hphantom{.}{$0.152$}     & \hphantom{.}{$0.286$}     & {$0.608$}             & \hphantom{0}{$3.58$}  & \hphantom{0}{$0.44$}  & {$3.975$}             & \hphantom{0.}{$0.94$} & \hphantom{0}{$0.06$}  & \hphantom{0}{$4.47$}  & \hphantom{.}$0.637$  \\

Freq. [\SI{}{\mega\hertz}]          & \hphantom{0}{$1414$}      & \hphantom{0}{$1255$}      & \hphantom{.0}$977$    & \hphantom{0.}{$637$}  & \hphantom{0.}{$721$}  & \hphantom{0.}{$722$}  & \hphantom{00}{$617$}  & \hphantom{.}{$1100$}  & \hphantom{0.}{$594$}  & \hphantom{00}$430$ \\

T/P [\SI{}{Gbps}]                   & \hphantom{.}{$4.457$}     & \hphantom{.}{$3.619$}     & $2.532$               & {$0.246$}             & {$0.307$}             & {$1.198$}             & \hphantom{.}{$1.608$} & {$1.426$}             & {$1.273$}             & \hphantom{0.}$3.25$   \\ \midrule

\mc{11}{l}{Scaled to \SI{28}{\nano\meter}$^\ddagger$}  \\
T/P [\SI{}{Gbps}]                   & \hphantom{.}{$4.457$}     & \hphantom{.}{$3.619$}     & $2.532$               & {$0.791$}             & {$0.307$}             & {$2.781$}             & \hphantom{.}{$3.224$} & {$0.815$}             & {$0.032$}             & \hphantom{.}$4.643$  \\

Area Eff. [\SI{}{Mbps\per\milli\meter\squared}]
                                    & {$29335$}                 & {$12671$}                 & \hphantom{.}$4165$    & \hphantom{.}{$2282$}  & \hphantom{0.}{$692$}  & \hphantom{.}{$3770$}  & {$19090$}             & \hphantom{.}{$4435$}  & \hphantom{0}{$73.89$}  & $14869$  \\ \bottomrule
 \end{tabular}
    \begin{tablenotes}
        \footnotesize
        \item[*] $^\dagger$ Our works are implemented with $|\mathbb{S}|_{\max}=4$ and compatible with all UL codewords, where latency is for UL-$(1024,512)$.
        \item[*] $^\ddagger$ Synthesis results are scaled to \SI{28}{\nano\meter} with area $\varpropto$ $s^2$ and frequency $\varpropto$ 1/$s$.
        \item[*] $^*$ Only the worst-case latency of these dynamic implementations is considered in the scaled results.
    \end{tablenotes}
\end{table*}

\section{Conclusions}\label{sec:conclusion}

In this work, we presented a low-latency and low-complexity generalized SR node-based SCL decoder implementation optimized for 5G NR polar codes.
Moreover, we applied several optimizations to significantly reduce the decoding latency and complexity of SR-List decoding to better meet the strict 5G requirements.
Numerical results show that our final decoder using all optimizations has a \SI{48.2}{\percent} reduction in decoding latency compared to the SOA decoder.
While some optimizations are specific to the SR node, many of our optimizations and analysis can be used with other types of polar decoders.
Additionally, our approach of simplifying the decoding process of the generalized SR node may also find use in other decoders for different types of generalized nodes.
We implemented the decoder and synthesis results based on a STM \SI{28}{\nano\meter} FD-SOI technology show that the proposed decoder for DL with maximum code length $512$ and $L=2$ yields a throughput of \SI{5.065}{Gbps} and an area efficiency of \SI{47.09}{Gbps\per\milli\meter\squared}.
For length-$1024$ UL polar codes, the SR-List decoder with $L=8$ achieves an area efficiency of \SI{4.165}{Gbps\per\milli\meter\squared} and a \SI{2.532}{Gbps} throughput which exceeds the SOA decoders.

\section*{Acknowledgments}
\pf{This work is supported by HiSilicon, Huawei Technologies Corporation.
The authors would like to thank the anonymous reviewers for their helpful comments.}

\section*{Dedication}

\pf{{This paper is dedicated to the memory of Prof. Alexander Vardy (deceased March 2022), who contributed significantly to the development of polar codes and several other areas of coding theory.}}

\bibliographystyle{IEEEtran}
\bibliography{IEEEabrv, bibliography}

\newpage

\appendix

\setcounter{table}{0}
\renewcommand{\thetable}{\Alph{section}\arabic{table}}

\subsection{Symbol and Function Definitions}\label{sec:symbol_def}

\begin{table}[ht]
  \captionsetup{font=small, justification=centering}
  \centering
  \caption{\pf{\MakeUppercase{Symbol and function definitions.}}}
  \label{tab:symbol_def}

  \small
  \renewcommand{\arraystretch}{1.1}

  \resizebox{!}{!}{\begin{tabular}{m{1.75cm} m{6.25cm}}
    \toprule
    \multicolumn{1}{l}{\textbf{Symbol}} & \multicolumn{1}{l}{\textbf{Definition}} \\

    \bottomrule \toprule

    $A$                     & Number of message bits. \\ \hline
    $E$                     & Codeword length after rate-matching. \\ \hline
    $G$                     & Encoded block length.  \\ \hline
    $K$                     & Number of information bits with CRC bits. \\ \hline
    $L$                     & List size. \\ \hline
    $N$                     & Mother polar code length. \\ \hline
    $R$                     & Code-rate $R$ with $R=A/E$. \\ \hline
    $P$                     & Number of CRC bits.  \\

    \bottomrule \toprule

    $\bbA$                  & Information bit set indices. \\ \hline
    $\bbF$                  & Frozen  bit set indices. \\

    \bottomrule \toprule

    $f(x,y)$                & $f$-function, defined in~\eqref{eq:fg_func_defA}. \\ \hline
    $g(x,y,z)$              & $g$-function, defined in~\eqref{eq:fg_func_defB}. \\ \hline

    $\HD(x)$                & Hard decision function $\HD(x) := 1_{x < 0}$. \\

    \bottomrule \toprule

    $\calN_{s,i}$           & $i$\=/th node in stage $s$ with code length $2^s = N_s$. \\ \hline

    $\PM_{s,i}^l$           & Path metric (PM) at $i$\=/th node in stage $s$ for path $l$, defined in~\eqref{eq:PM_SCL_func}. \\ \hline

    $\psum_{s,i}^{l}$       & Partial sum (PSUM) vector at $i$\=/th node in stage $s$ for path $l$, defined in~\eqref{eq:psum_func}. $\widetilde{\psum}$ denotes the hard decision PSUM on the SR node's input LLR and $\widehat{\psum}$ the maximum likelihood (ML) PSUM. \\ \hline

    $\gamma_q$              & Parity of the $q$\=/th group of source node partial sums (PSUMs) as defined in~\eqref{eq:gpc_parity_SCL}. \\ \hline

    $\Delta_{r,j}^l$        & Path metric (PM) penalty for path $l$ at node $\calN_{r,j}$ defined in~\eqref{eq:Delta_PM_gpc_SCL}. \\ \hline

    $\Delta^{l,k}$          & Path metric (PM) penalty from the parity constraint of the source node (G-PC node) for path $l$ with repetition sequence $k$ defined in~\eqref{eq:PM_SRI_SCL}. \\ \hline

    $\epsilon_q $           & Index of the minimum magnitude LLR in the $q$\=/th group of source node LLRs, defined in~\eqref{eq:gpc_index_SCL}. \\ \hline

    ${\bm{\zeta}}_{r}^{l}$  & Modified LLRs for each sub-group with increased or decreased magnitude when parity constraint is satisfied or not, defined in~\eqref{eq:llr_gpc_modified_SCL}. \\ \hline

    $\llr_{s,i}^{l}$        & LLR vector at $i$\=/th node in stage $s$ for path $l$, defined in~\eqref{eq:llr_fg_func}. $\llr^k$ is the LLR to the source node of an SR node for the $k$\=/th repetition sequence. \\

    \bottomrule \toprule

    $\mathrm{SR}(\bv, \mathrm{SNT}, r)$   & Sequence repetition (SR) node parametrization. \\ \hline

    $r$                     & Stage of the source node, $\calN_r$. \\ \hline

    $\bbS$                  & $\bbS = \{\bbS^0, \bbS^1, \dots, \bbS^{2^{W_v}-1}\}$ is the repetition sequence set, defined in \eqref{eq:seq_rep_def}. \\ \hline

    $\mathrm{SNT}$          & \textbf{S}pecial \textbf{N}ode \textbf{T}ype of the source node.  \\ \hline

    $T_\mathrm{SNT}$        & Upper threshold on number of path forks for a source node of type $\mathrm{SNT}$, defined in~\eqref{eq:empirical_path}. \\ \hline

    $\bv$                   & Binary vector indicating the R0 and REP node distribution from $\calN_s$ to $\calN_r$. $\bv[t] = 0$ indicates an R0 node and $\bv[t] = 1$ an REP node. \\ \hline

    $W_v$                   & Sum of $\bv$ (number of REP nodes in the SR node). Equal to cardinality of $\bbS$. \\ \hline

    $\bm{\eta}$             & Last value of each R0/REP node, that is, $\bm{\eta}[t]=0$ when $\bv[t]=0$ and $\bm{\eta}[t]\in\{0,1\}$ when $\bv[t]=1$. \\

    \bottomrule
  \end{tabular}}
\end{table}

\end{document}

%% file: tikz/colorscheme.tex
\definecolor{bostonunired}{HTML}{CC0000}    %
\definecolor{fireenginered}{HTML}{CF202A}   %
\definecolor{vermillionred}{HTML}{E34234}   %

\definecolor{titlered}{RGB}{212,0,0}

\definecolor{reddeep}{RGB}{178,24,43}
\definecolor{redlight}{RGB}{252,78,42}
\definecolor{redlight2}{RGB}{255,114,111}

\definecolor{redbright}{RGB}{255,24,43}

\definecolor{perfectgreen}{HTML}{4FBF26}    %
\definecolor{maygreen}{HTML}{4E9B47}        %
\definecolor{mattelime}{HTML}{75AD50}       %

\definecolor{indigorainbow}{HTML}{1D3F6E}   %
\definecolor{royalazure}{HTML}{1866E1}      %
\definecolor{azure}{HTML}{008AFF}           %
\definecolor{blueberry}{HTML}{4F86F7}       %
\definecolor{fadednavy}{HTML}{242F78}       %
\definecolor{navy}{HTML}{000080}            %

\definecolor{goldwebgolden}{HTML}{FFD700}   %
\definecolor{radioactive}{HTML}{FAE500}     %
\definecolor{gold}{HTML}{FFD700}            %

\definecolor{vividgamboge}{HTML}{FF9900} %

\definecolor{shinygray}{HTML}{C7C6C6}

\definecolor{slategrey}{HTML}{708090}   %
\definecolor{neutralgray}{HTML}{828382} %

\definecolor{mattecharcoal}{HTML}{3B4248} %

\definecolor{graydark}{HTML}{6B6E70} %

\definecolor{charcoal}{HTML}{36454F} %
\definecolor{charcoalShade1}{HTML}{2B373F}
\definecolor{charcoalShade2}{HTML}{263037}
\definecolor{charcoalShade3}{HTML}{20292F}

\definecolor{sunset1}{HTML}{F3E79B}
\definecolor{sunset2}{HTML}{FAC484}
\definecolor{sunset3}{HTML}{F8A07E}
\definecolor{sunset4}{HTML}{EB7F86}
\definecolor{sunset5}{HTML}{CE6693}
\definecolor{sunset6}{HTML}{A059A0}
\definecolor{sunset7}{HTML}{5C53A5}

\definecolor{bluyi1}{HTML}{F7FEAE}
\definecolor{bluyi2}{HTML}{B7E6A5}
\definecolor{bluyi3}{HTML}{7CCBA2}
\definecolor{bluyi4}{HTML}{46AEA0}
\definecolor{bluyi5}{HTML}{089099}
\definecolor{bluyi6}{HTML}{00718B}
\definecolor{bluyi7}{HTML}{045275}

\definecolor{geyser1}{HTML}{008080}
\definecolor{geyser2}{HTML}{70A494}
\definecolor{geyser3}{HTML}{B4C8A8}
\definecolor{geyser4}{HTML}{F6EDBD}
\definecolor{geyser5}{HTML}{EDBB8A}
\definecolor{geyser6}{HTML}{DE8A5A}
\definecolor{geyser7}{HTML}{CA562C}

\definecolor{temps1}{HTML}{009392}
\definecolor{temps2}{HTML}{39B185}
\definecolor{temps3}{HTML}{9CCb86}
\definecolor{temps4}{HTML}{E9E29C}
\definecolor{temps5}{HTML}{EEB479}
\definecolor{temps6}{HTML}{E88471}
\definecolor{temps7}{HTML}{CF597E}

\definecolor{earth1}{HTML}{A16928}
\definecolor{earth2}{HTML}{BD925A}
\definecolor{earth3}{HTML}{D6BD8D}
\definecolor{earth4}{HTML}{EDEAC2}
\definecolor{earth5}{HTML}{B5C8B8}
\definecolor{earth6}{HTML}{79A7AC}
\definecolor{earth7}{HTML}{2887A1}

\definecolor{bold1}{HTML}{7F3C8D}
\definecolor{bold2}{HTML}{11A579}
\definecolor{bold3}{HTML}{3969AC}
\definecolor{bold4}{HTML}{F2B701}
\definecolor{bold5}{HTML}{E73F74}
\definecolor{bold6}{HTML}{80BA5A}
\definecolor{bold7}{HTML}{E68310}
\definecolor{bold8}{HTML}{008695}
\definecolor{bold9}{HTML}{CF1C90}
\definecolor{bold10}{HTML}{F97B72}
\definecolor{bold11}{HTML}{4B4B8F}
\definecolor{bold12}{HTML}{A5AA99}

\definecolor{pastel1}{HTML}{66C5CC}
\definecolor{pastel2}{HTML}{F6CF71}
\definecolor{pastel3}{HTML}{F89C74}
\definecolor{pastel4}{HTML}{DCB0F2}
\definecolor{pastel5}{HTML}{87C55F}
\definecolor{pastel6}{HTML}{9EB9F3}
\definecolor{pastel7}{HTML}{FE88B1}
\definecolor{pastel8}{HTML}{C9DB74}
\definecolor{pastel9}{HTML}{8BE0A4}
\definecolor{pastel10}{HTML}{B497E7}
\definecolor{pastel11}{HTML}{D3B484}
\definecolor{pastel12}{HTML}{B3B3B3}

\definecolor{prism1}{HTML}{5F4690}
\definecolor{prism2}{HTML}{1D6996}
\definecolor{prism3}{HTML}{38A6A5}
\definecolor{prism4}{HTML}{0F8554}
\definecolor{prism5}{HTML}{73AF48}
\definecolor{prism6}{HTML}{EDAD08}
\definecolor{prism7}{HTML}{E17C05}
\definecolor{prism8}{HTML}{CC503E}
\definecolor{prism9}{HTML}{94346E}
\definecolor{prism10}{HTML}{6F4070}
\definecolor{prism11}{HTML}{994E95}
\definecolor{prism12}{HTML}{666666}

\definecolor{vivid1}{HTML}{E58606}
\definecolor{vivid2}{HTML}{5D69B1}
\definecolor{vivid3}{HTML}{52BCA3}
\definecolor{vivid4}{HTML}{99C945}
\definecolor{vivid5}{HTML}{CC61B0}
\definecolor{vivid6}{HTML}{24796C}
\definecolor{vivid7}{HTML}{DAA51B}
\definecolor{vivid8}{HTML}{2F8AC4}
\definecolor{vivid9}{HTML}{764E9F}
\definecolor{vivid10}{HTML}{ED645A}
\definecolor{vivid11}{HTML}{CC3A8E}
\definecolor{vivid12}{HTML}{A5AA99}

\definecolor{cBrewerPaired1}{HTML}{A6CEE3}
\definecolor{cBrewerPaired2}{HTML}{1F78B4}
\definecolor{cBrewerPaired3}{HTML}{B2DF8A}
\definecolor{cBrewerPaired4}{HTML}{33A02C}
\definecolor{cBrewerPaired5}{HTML}{FB9A99}
\definecolor{cBrewerPaired6}{HTML}{E31A1C}
\definecolor{cBrewerPaired7}{HTML}{FDBF6F}
\definecolor{cBrewerPaired8}{HTML}{FF7F00}
\definecolor{cBrewerPaired9}{HTML}{CAB2D6}
\definecolor{cBrewerPaired10}{HTML}{6A3D9A}
\definecolor{cBrewerPaired11}{HTML}{FFFF99}
\definecolor{cBrewerPaired12}{HTML}{B15928}

\definecolor{cBrewerQualPrint1}{HTML}{E41A1C}
\definecolor{cBrewerQualPrint2}{HTML}{377EB8}
\definecolor{cBrewerQualPrint3}{HTML}{4DAF4A}
\definecolor{cBrewerQualPrint4}{HTML}{984EA3}
\definecolor{cBrewerQualPrint5}{HTML}{FF7F00}
\definecolor{cBrewerQualPrint6}{HTML}{FFFF33}
\definecolor{cBrewerQualPrint7}{HTML}{A65628}
\definecolor{cBrewerQualPrint8}{HTML}{F781BF}
\definecolor{cBrewerQualPrint9}{HTML}{999999}

\definecolor{thsViolet1}{HTML}{E4C7F1} %
\definecolor{thsViolet2}{HTML}{9F82CE} %
\definecolor{thsViolet3}{HTML}{4B4B8F} %

\definecolor{thsPink1}{HTML}{DCB0F2} %
\definecolor{thsPink2}{HTML}{FE88B1} %
\definecolor{thsPink3}{HTML}{CC3A8E} %

\definecolor{thsBlue1}{HTML}{4F86F7} %
\definecolor{thsBlue2}{HTML}{0F52BA} %
\definecolor{thsBlue3}{HTML}{332288} %

\definecolor{thsCyan1}{HTML}{88CCEE} %
\definecolor{thsCyan2}{HTML}{66C5CC} %
\definecolor{thsCyan3}{HTML}{367588} %

\definecolor{thsGreen1}{HTML}{4FBF26} %
\definecolor{thsGreen2}{HTML}{11A579} %
\definecolor{thsGreen3}{HTML}{0F8554} %

\definecolor{thsYellow1}{HTML}{FFFF66} %
\definecolor{thsYellow2}{HTML}{FAE500} %
\definecolor{thsYellow2}{HTML}{FFFF00}

\definecolor{thsOrange1}{HTML}{ECDA9A} %
\definecolor{thsOrange2}{HTML}{F7945D} %
\definecolor{thsOrange3}{HTML}{E58606} %

\definecolor{thsRed1}{HTML}{CC503E} %
\definecolor{thsRed2}{HTML}{CC0000} %
\definecolor{thsRed3}{HTML}{670E10} %

\definecolor{thsGray1}{HTML}{B1B6B7} %
\definecolor{thsGray2}{HTML}{828382} %
\definecolor{thsGray3}{HTML}{36454F} %

\def\RZeroColor{white}
\def\ROneColor{black}
\def\RepColor{gold}
\def\SpcColor{perfectgreen}

\def\SRColor{cyan!70!white}
\def\SourceColor{redlight2}

\newcommand{\redT}[1]{\textcolor{vermillionred}{#1}}

\newcommand{\blueT}[1]{\textcolor{blue}{#1}}
\newcommand{\greenT}[1]{\textcolor{\SpcColor}{#1}}